\documentclass[11pt]{article} 
	\pdfoutput=1
	\usepackage{tikz}
	\usepackage{amsfonts, amsmath, amssymb, amsthm}
\usepackage{bbm}
\usepackage{thmtools,thm-restate}
\usepackage{algpseudocode}
\usepackage{graphicx}
\usepackage{stmaryrd}
\usepackage[show]{ed}
\usepackage{version}
\usepackage[small]{complexity}
\usepackage[normalem]{ulem}
\usepackage{xspace}
\usepackage{lmodern}
\usepackage{nicefrac}
\usepackage[noadjust]{cite}
\usepackage[margin=1in]{geometry}
\usepackage{microtype}
\usepackage{nag}
\usepackage{nth}
\usepackage{intcalc}
\usepackage{etoolbox}
\usepackage{tocbibind}
\usepackage[hyphens]{url}
\usepackage[pdftex,colorlinks=true,linkcolor=blue,citecolor=magenta]{hyperref}
\usepackage{algorithm}

\SetSymbolFont{stmry}{bold}{U}{stmry}{m}{n}

\makeatletter
\global\let\tikz@ensure@dollar@catcode=\relax
\makeatother

\makeatletter
\renewcommand\thmt@listnumwidth{4.3em}
\makeatother

\numberwithin{equation}{section}
\declaretheoremstyle[bodyfont=\it,qed=$\qedsymbol$]{noproofstyle} 

\theoremstyle{plain}
\declaretheorem[numberlike=equation]{theorem}
\declaretheorem[unnumbered,name=Theorem]{theorem*}

\declaretheorem[unnumbered,name=Theorem,style=noproofstyle]{theoremwp*}

\declaretheorem[numberlike=equation]{lemma}
\declaretheorem[unnumbered,name=Lemma]{lemma*}
\declaretheorem[numberlike=equation,name=Lemma,style=noproofstyle]{lemmawp}
\declaretheorem[unnumbered,name=Lemma,style=noproofstyle]{lemmawp*}

\declaretheorem[numberlike=equation]{corollary}
\declaretheorem[unnumbered,name=Corollary]{corollary*}
\declaretheorem[numberlike=equation,name=Corollary,style=noproofstyle]{corollarywp}
\declaretheorem[unnumbered,name=Corollary,style=noproofstyle]{corollarywp*}

\declaretheorem[numberlike=equation]{proposition}
\declaretheorem[unnumbered,name=Proposition]{proposition*}
\declaretheorem[numberlike=equation,name=Proposition,style=noproofstyle]{propositionwp}
\declaretheorem[unnumbered,name=Proposition,style=noproofstyle]{propositionwp*}

\declaretheorem[unnumbered,name=Claim]{claim*}

\declaretheorem[unnumbered,name=Conjecture]{conjecture*}

\declaretheorem[numberlike=equation,style=noproofstyle]{construction}
\declaretheorem[unnumbered,name=Construction=noproofstyle]{construction*}

\declaretheorem[numberlike=equation,name=Theorem,style=noproofstyle]{theorem-cited}

\declaretheoremstyle[qed=$\lozenge$,bodyfont=\it]{defstyle} 
\declaretheorem[numberlike=equation,style=defstyle]{definition}
\declaretheorem[unnumbered,name=Definition,style=defstyle]{definition*}

\declaretheorem[unnumbered,name=Example,style=defstyle]{example*}

\declaretheorem[unnumbered,name=Notation,style=defstyle]{notation*}

\declaretheorem[unnumbered,name=Question,style=defstyle]{question*}

\declaretheoremstyle[qed=$\lozenge$]{rmkstyle} 
\declaretheorem[numberlike=equation,style=rmkstyle]{remark}
\declaretheorem[unnumbered,name=Remark,style=rmkstyle]{remark*}

\newcommand{\demph}[1]{\textbf{#1}}

\newcommand{\ocint}[2]{\ensuremath{(#1,#2]}}
\newcommand{\coint}[2]{\ensuremath{[#1,#2)}}
\newcommand{\zr}[1]{{\llbracket {#1}\rrbracket}}

\newcommand{\bits}{\ensuremath{\{0,1\}}}

\newcommand{\la}{\langle}
\newcommand{\ra}{\rangle}

\DeclareMathOperator{\rank}{rank}

\DeclareMathOperator{\spn}{span}

\DeclareMathOperator{\cspn}{col-span}
\DeclareMathOperator{\rspn}{row-span}

\DeclareMathOperator{\tr}{tr}

\DeclareMathOperator{\GL}{GL}
\let\SL\undefined
\DeclareMathOperator{\SL}{SL}

\newcommand{\st}{\ensuremath{\ |\ }}

\newcommand{\floor}[1]{{\lfloor{#1}\rfloor}}

\newcommand{\ceil}[1]{{\lceil{#1}\rceil}}	
\newcommand{\ceilp}[1]{{\left\lceil{#1}\right\rceil}}

\newcommand{\F}{\mathbb{F}}

\renewcommand{\K}{\mathbb{K}}
\renewcommand{\R}{\mathbb{R}}
\newcommand{\N}{\mathbb{N}}

\newcommand{\Z}{\mathbb{Z}}

\newcommand{\e}{\mathrm{e}}
\newcommand{\T}{\mathrm{tr}}
\newcommand{\eps}{\epsilon}

\newcommand{\eqdef}{:=}

\renewcommand{\O}{O}
\newcommand{\ignore}[1]{}

\newcommand{\Id}{\mathrm{I}}

\newcommand{\rv}{\mathsf}

\newcommand{\cA}{\mathcal{A}}

\newcommand{\cC}{\mathcal{C}}

\newcommand{\cE}{\mathcal{E}}

\newcommand{\cH}{\mathcal{H}}

\newcommand{\cR}{\mathcal{R}}

\renewcommand{\vec}[1]{\overline{#1}}

\newcommand{\rA}{\rv{A}}

\newcommand{\rE}{\rv{E}}

\newcommand{\rM}{\rv{M}}
\newcommand{\rN}{\rv{N}}

\makeatletter
\newcommand{\va}{{\vec{a}}\@ifnextchar{^}{\!\:}{}}
\newcommand{\vb}{{\vec{b}}\@ifnextchar{^}{\!\:}{}}
\newcommand{\vc}{{\vec{c}}\@ifnextchar{^}{\!\:}{}}
\newcommand{\vd}{{\vec{d}}\@ifnextchar{^}{\!\:}{}}
\newcommand{\ve}{{\vec{e}}\@ifnextchar{^}{\!\:}{}}

\newcommand{\vg}{{\vec{g}}\@ifnextchar{^}{\!\:}{}}
\newcommand{\vh}{{\vec{h}}\@ifnextchar{^}{\!\:}{}}
\newcommand{\vi}{{\vec{i}}\@ifnextchar{^}{\!\:}{}}
\newcommand{\vj}{{\vec{j}}\@ifnextchar{^}{\!\:}{}}
\newcommand{\vk}{{\vec{k}}\@ifnextchar{^}{\!\:}{}}
\newcommand{\vl}{{\vec{\ell}}\@ifnextchar{^}{\!\:}{}}
\newcommand{\vm}{{\vec{m}}\@ifnextchar{^}{\!\:}{}}
\newcommand{\vn}{{\vec{n}}\@ifnextchar{^}{\!\:}{}}
\newcommand{\vo}{{\vec{o}}\@ifnextchar{^}{\!\:}{}}
\newcommand{\vp}{{\vec{p}}\@ifnextchar{^}{\!\:}{}}
\newcommand{\vq}{{\vec{q}}\@ifnextchar{^}{\!\:}{}}
\newcommand{\vr}{{\vec{r}}\@ifnextchar{^}{\!\:}{}}
\newcommand{\vs}{{\vec{s}}\@ifnextchar{^}{\!\:}{}}
\newcommand{\vt}{{\vec{t}}\@ifnextchar{^}{\!\:}{}}
\newcommand{\vu}{{\vec{u}}\@ifnextchar{^}{\!\:}{}}
\newcommand{\vv}{{\vec{v}}\@ifnextchar{^}{\!\:}{}}
\newcommand{\vw}{{\vec{w}}\@ifnextchar{^}{\!\:}{}}
\newcommand{\vy}{{\vec{y}}\@ifnextchar{^}{\!\:}{}}
\newcommand{\vx}{{\vec{x}}\@ifnextchar{^}{}{}}		
\newcommand{\vz}{{\vec{z}}\@ifnextchar{^}{\!\:}{}}

\newcommand{\vA}{{\vec{A}}\@ifnextchar{^}{\!\:}{}}
\newcommand{\vB}{{\vec{B}}\@ifnextchar{^}{\!\:}{}}
\newcommand{\vC}{{\vec{C}}\@ifnextchar{^}{\!\:}{}}
\newcommand{\vD}{{\vec{D}}\@ifnextchar{^}{\!\:}{}}
\newcommand{\vE}{{\vec{E}}\@ifnextchar{^}{\!\:}{}}
\newcommand{\vF}{{\vec{F}}\@ifnextchar{^}{\!\:}{}}
\newcommand{\vG}{{\vec{G}}\@ifnextchar{^}{\!\:}{}}
\newcommand{\vH}{{\vec{H}}\@ifnextchar{^}{\!\:}{}}
\newcommand{\vI}{{\vec{I}}\@ifnextchar{^}{\!\:}{}}
\newcommand{\vJ}{{\vec{J}}\@ifnextchar{^}{\!\:}{}}
\newcommand{\vK}{{\vec{K}}\@ifnextchar{^}{\!\:}{}}
\newcommand{\vL}{{\vec{L}}\@ifnextchar{^}{\!\:}{}}
\newcommand{\vM}{{\vec{M}}\@ifnextchar{^}{\!\:}{}}
\newcommand{\vN}{{\vec{N}}\@ifnextchar{^}{\!\:}{}}
\newcommand{\vO}{{\vec{O}}\@ifnextchar{^}{\!\:}{}}
\newcommand{\vP}{{\vec{P}}\@ifnextchar{^}{\!\:}{}}
\newcommand{\vQ}{{\vec{Q}}\@ifnextchar{^}{\!\:}{}}
\newcommand{\vR}{{\vec{R}}\@ifnextchar{^}{\!\:}{}}
\newcommand{\vS}{{\vec{S}}\@ifnextchar{^}{\!\:}{}}
\newcommand{\vT}{{\vec{T}}\@ifnextchar{^}{\!\:}{}}
\newcommand{\vU}{{\vec{U}}\@ifnextchar{^}{\!\:}{}}
\newcommand{\vV}{{\vec{V}}\@ifnextchar{^}{\!\:}{}}
\newcommand{\vW}{{\vec{W}}\@ifnextchar{^}{\!\:}{}}
\newcommand{\vY}{{\vec{Y}}\@ifnextchar{^}{\!\:}{}}
\newcommand{\vX}{{\vec{X}}\@ifnextchar{^}{}{}}		
\newcommand{\vZ}{{\vec{Z}}\@ifnextchar{^}{\!\:}{}}
\makeatother

\newcommand{\vaa}{{\vec{\alpha}}}

	\usepackage{mathptmx}
	\DeclareMathAlphabet{\mathcal}{OMS}{cmsy}{m}{n}

	\renewcommand{\le}{\leqslant}
	
	\renewcommand{\ge}{\geqslant}
	\DeclareMathOperator{\Wr}{W}

	\title{Dimension Expanders via Rank Condensers}
	\author{%
		Michael A.\ Forbes
		\thanks{
			Email: \texttt{miforbes@csail.mit.edu}.
			Simons Institute for the Theory of Computing,
			Calvin Lab, UC Berkeley
			Berkeley, CA 94720-2190.
			This work was performed when the author was a graduate student at MIT CSAIL (which was supported by the Center for Science of Information (CSoI), a NSF Science and Technology Center, under grant agreement CCF-0939370) and when the author was a Google Research Fellow at the Simons Institute for the Theory of Computing.
		}
		\and
		Venkatesan Guruswami
		\thanks{
			Email: \texttt{guruswami@cmu.edu}.
			Computer Science Department,
			Carnegie Mellon University,
			Pittsburgh, PA.
			Some of this work was done when the author was a visiting researcher at Microsoft Research New England, Cambridge, MA. Research supported in part by NSF grant CCF-0963975.
		}
	}

\begin{document}
\maketitle
\thispagestyle{empty}

\begin{abstract}
	An emerging theory of ``linear-algebraic pseudorandomness'' aims to understand the linear-algebraic analogs of fundamental Boolean pseudorandom objects where the rank of subspaces plays the role of the size of subsets.  In this work, we study and highlight the interrelationships between several such algebraic objects such as subspace designs, dimension expanders, \emph{seeded rank condensers}, \emph{two-source rank condensers}, and rank-metric codes.  In particular, with the recent construction of near-optimal subspace designs by Guruswami and Kopparty~\cite{GuruswamiKopparty13} as a starting point, we construct good (seeded) rank condensers (both \emph{lossless} and \emph{lossy} versions), which are a small collection of linear maps $\F^n \to \F^t$ for $t \ll n$ such that for every subset of $\F^n$ of small rank, its rank is preserved (up to a constant factor in the lossy case) by at least one of the maps. 

	We then compose a tensoring operation with our lossy rank condenser to construct constant-degree dimension expanders over polynomially large fields. That is, we give $O(1)$ explicit linear maps $A_i : \F^n \to \F^n$ such that for any subspace $V \subseteq \F^n$ of dimension at most $n/2$, $\dim\bigl( \sum_i A_i(V)\bigr) \ge (1+\Omega(1)) \dim(V)$. Previous constructions of such constant-degree dimension expanders were based on Kazhdan's property $T$ (for the case when $\F$ has characteristic zero) or monotone expanders (for every field $\F$); in either case the construction was {\em harder} than that of usual vertex expanders. Our construction, on the other hand, is {\em simpler}.

	For two-source rank condensers, we observe that the lossless variant (where the output rank is the product of the ranks of the two sources) is equivalent to the notion of a linear rank-metric code. For the lossy case, using our seeded rank condensers, we give a reduction of the general problem to the case when the sources have high ($n^{\Omega(1)}$) rank. When the sources have $O(1)$ rank, combining this with an ``inner condenser'' found by brute-force leads to a two-source rank condenser with output length nearly matching the probabilistic constructions.
\end{abstract}

\newpage
\enlargethispage{0.5cm}
\parskip=0.1ex
\tableofcontents
\thispagestyle{empty}
\newpage
\parskip=0.5ex

\section{Introduction}

The broad area of pseudorandomness deals with efficiently generating objects that exhibit the desirable properties of ``random-like'' objects despite being constructed either explicitly or with limited randomness. Pseudorandomness is a central and influential theme in many areas such as complexity theory, derandomization, coding theory, cryptography, high-dimensional geometry, graph theory, and additive combinatorics. The topic has witnessed much progress over the years and continues to be intensively studied. We now have non-trivial constructions of various pseudorandom objects such as expander graphs, randomness extractors and condensers, Ramsey graphs, list-decodable codes, compressed sensing matrices, Euclidean sections, and pseudorandom generators for various concrete models.  Despite the seemingly different definitions and contexts of these objects, insights in pseudorandomness have uncovered intimate connections between them, and this has led to a rich theory of ``Boolean pseudorandomness'' drawing a common pool of broadly useful techniques (see for instance the recent survey by Vadhan~\cite{Vadhan12}.)

Recently, there is an emerging theory of ``algebraic pseudorandomness'' aimed at understanding the linear-algebraic analogs of fundamental Boolean pseudorandom objects where the dimension of subspaces plays the role analogous to min-entropy. Examples of such algebraic objects include dimension expanders, subspace-evasive sets, subspace designs, rank-preserving condensers, etc. In addition to their intrinsic interest, these notions also have surprising applications; for instance, subspace-evasive sets to the construction of Ramsey graphs~\cite{PudlakRodl04} and list-decodable codes~\cite{GuruswamiWang13,GuruswamiXing12}, subspace designs to list decoding both in the Hamming metric and the rank metric~\cite{GuruswamiXing13,GuruswamiWang14}, and rank-preserving condensers to affine extractors~\cite{GabizonRaz08} and polynomial identity testing~\cite{KarninShpilka11,ForbesShpilka12}. 

In this work, we study several interesting pseudorandom objects in the linear-algebraic world, such as subspace evasive sets, subspace designs, dimension expanders, \emph{seeded rank condensers}, and \emph{two-source rank condensers}. The last two notions are also introduced in this work, though closely related concepts were studied earlier in the literature.  We briefly and informally define these notions now, with more precise statements appearing in later sections.  A subspace evasive set is a (large) subset of $\F^n$ that has small intersection with every low-dimensional subspace of $\F^n$. Subspace designs are a (large) collection of subspaces such that every low-dimensional subspace intersects few of them. Dimension expanders are a (small) collection of linear maps $A_i: \F^n \to \F^n$ such that for every subspace $V \subseteq \F^n$ of bounded dimension, the dimension of $\sum_i A_i(V)$ is at least $\alpha\cdot \dim(V)$ for a constant $\alpha > 1$.  Rank condensers are a (small) collection of linear maps $\F^n \to \F^t$ (for $t \ll n$) such that for every subspace of dimension $r$, its image under at least one of the maps has large dimension (equal to $r$ in the \emph{lossless} case, and $\Omega(r)$ in the \emph{lossy} case).  A two-source rank condenser is a map $E: \F^n \times \F^n \to \F^t$ such that for every pair $A,B \subseteq \F^n$ with rank $r$ each, $f(A \times B)$ has rank $\Omega(r^2)$ (or even $r^2$ in the lossless case) --- the tensor product construction is lossless but requires $t = n^2$, so the challenge here is to ``derandomize'' the tensor product and achieve $t \ll n^2$ (and even $t \ll n$ for the lossy case for $r\ll\sqrt{n}$).

Conceptually, our work highlights close interconnections between these notions. In particular, we show that subspace designs (which were introduced in the context of list decoding variants of algebraic-geometric codes in \cite{GuruswamiXing13}) are the {\em same} concept as lossless rank condensers but that they emphasize a different regime of parameters. This connection also highlights that a strong variant of subspace designs yields lossy rank condensers. The near-optimal explicit construction of (strong) subspace designs in \cite{GuruswamiKopparty13} then yields lossless and lossy rank condensers with parameters close to the existential constructions.  Our main technical application is an explicit construction of constant-degree dimension expanders over polynomially large fields, that expands all subspaces of $\F^n$ of dimension $n/2$ (say) by a factor $\alpha > 1$. We achieve this construction by first increasing the rank in a trivial way by increasing the dimension of the ambient space, and then using a lossy rank condenser to reduce the ambient space back to $\F^n$ while preserving the rank up to a constant factor. While previous constructions of dimension expanders were at least as complicated as constructions of standard expander graphs (or more so), our construction and analysis is rather elementary.  Unfortunately, unlike previous work, our techniques are currently best suited to large fields due to connections with Reed-Solomon codes. However, we do obtain dimension expanders over small fields by paying various logarithmic penalties.

Turning to two-source rank condensers, our original motivation to propose them was a possible route to iteratively construct subspace-evasive sets that might offer some way around the exponential dependence on intersection size that seems inherent to constructions based on algebraic varieties. While there appears to be serious obstacles to such an approach, the notion seems a fundamental one to study regardless.  In this work, we focus on two-source rank condensers $f : \F^n \times \F^n \to \F^t$ where the map $f$ is bilinear as this seems like a natural class of constructions to study. We observe that the lossless variant is {\em equivalent} to the notion of a linear rank-metric code. Known optimal constructions of rank-metric codes such as the Gabidulin codes thereby yield lossless two-source condensers with optimal output length (equal to $\Theta(nr)$ for rank-$r$ subsets of $\F^n$). For lossy two-source rank condensers, we can enumerate over the seeds of our seeded lossy condenser, applying it to both sources separately and condensing the sources to $r^{\Theta(1)}$ dimensions (from the original $n$).  For small $r$ (e.g., constant), we can ``concatenate'' this construction with a near-optimal lossy two-source condenser found by brute-force to obtain output length $\Theta(n/r)$, matching the non-constructive bound.  In general, our method reduces the problem to the case of relatively high ``rate'' (when $r \approx n^{1/3}$), which is typically easier to tackle.

\paragraph{Organization:} In the next three sections, we state (informal versions of) our results, all of ideas behind them, and brief discussions of prior work for seeded rank condensers (\autoref{sec:lossless-seeded}), dimension expanders (\autoref{sec:dim-exp}), and two-source rank condensers (\autoref{sec:two-src}). We then give an expanded treatment with formal statements and proofs, as well as detailed comparison with related work in Sections \ref{sec:lossless-seeded_constr}, \ref{sec:dim-exp_constr}, \ref{sec:small-fields} and \ref{sec:two-src_constr}. The connection between rank-metric codes and two-source lossless rank condensers is described in \autoref{sec:rank-metric}. We also include detailed results on the parameters achieved by random constructions in \autoref{sec:prob-method}.

\section{Subspace Designs and Rank Condensers}\label{sec:lossless-seeded}

We begin by discussing the notion of a \emph{subspace design}, as recently defined by Guruswami and Xing~\cite{GuruswamiXing13}, and contrast this with the notion of a \emph{seeded (single source) rank condenser} to which we add the qualifier of \emph{lossless}, as defined by Forbes, Saptharishi and Shpilka~\cite{ForbesSS14}.  We will describe how these objects are essentially the same notion, where the rank condenser can be considered the ``primal'' object and the subspace design the ``dual'' object.  We then introduce \emph{lossy rank condensers}, a new notion that is key to our construction of dimension expanders (see \autoref{sec:dim-exp}) and describe how the construction of subspace designs of Guruswami and Kopparty~\cite{GuruswamiKopparty13} implies nearly optimal lossy rank condensers.

\subsection{Subspace Designs}

We begin with the definition of a subspace design.

\begin{definition}[Guruswami-Xing~\cite{GuruswamiXing13} and Guruswami-Kopparty~\cite{GuruswamiKopparty13}]
	Let $\F$ be a field. A collection $\cH=\{H_i\}_i$ of subspaces $H_i\subseteq\F^n$ is a \demph{weak $(r,L)$-subspace design} if for every subspace $V\subseteq\F^n$ with $\dim V=r$, 
	\[
		|\{i\st \dim (H_i\cap V)>0\}|\le L
		\;.
	\]
	The collection $\cH$ is a \demph{strong $(r,L)$-subspace design} if for every subspace $V\subseteq\F^n$ with $\dim V=r$,
	\[
		\sum_i \dim (H_i\cap V)\le L
		\;.
	\]
	The collection $\cH$ is \demph{explicit} if given an index $i\in[|\cH|]$ a basis for the $i$-th subspace in $\cH$ can be constructed in $\poly(n,\log|\cH|)$ operations in $\F$.
\end{definition}

We note here that the above subspaces $H_i$ are not constrained to be of equal dimension.  Allowing the dimension of the $H_i$ to vary could conceivably allow for improved constructions, but no construction so far uses this freedom.  As such, we will primarily concern ourselves with the case when the dimensions are equal.

Guruswami-Xing~\cite{GuruswamiXing13} defined subspace designs as a way to prune list-decodable codes to ensure a small list-size while maintaining high rate.  As such, one wishes for the size $|\cH|$ of the design to be large while maintaining $L$ of moderate size.  In particular, they showed that large designs exist non-constructively.

\begin{propositionwp*}[Guruswami-Xing~\cite{GuruswamiXing13}]
	Let $\F_q$ be a finite field. Let $\eps>0$, $n\ge \nicefrac{8}{\eps}$ and $s\le\nicefrac{\eps n}{2}$. Then there is a strong $(s,\nicefrac{8s}{\eps})$-subspace design $\cH$ of $(1-\eps)n$-dimensional subspaces in $\F_q^n$ with $|\cH|=q^{\nicefrac{\eps n}{8}}$.
\end{propositionwp*}

Note that the \emph{co}-dimension of the subspaces in $\cH$ is $\eps n$, which is twice that of the maximum dimension $s\approx \nicefrac{\eps n}{2}$.  We now further remark on the variations of this definition. The following relation between the weak and strong versions is immediate.

\begin{lemmawp}[Guruswami-Kopparty~\cite{GuruswamiKopparty13}]\label{res:strong-v-weak}
	Let $\F$ be a field, and let $\cH$ be a collection of subspaces in $\F^n$.  Then
		if $\cH$ is a strong $(r,L)$-subspace design, then $\cH$ is a weak $(r,L)$-subspace design.
		If $\cH$ is a weak $(r,L)$-subspace design, then $\cH$ is a strong $(r,rL)$-subspace design.
\end{lemmawp}

We also observe that as every dimension $\le r$ subspace can be padded to a dimension $r$ subspace, we immediately can see that subspace designs apply to smaller subspaces as well.

\begin{lemmawp}\label{res:subspace design_les-vs-eqs}
	Let $\F$ be a field, and let $\cH$ be a weak/strong $(r,L)$-subspace design in $\F^n$.  Then $\cH$ is a $(s,L)$-subspace design over $\F^n$ for every $1\le s\le r$.
\end{lemmawp}

While the above seems to allow one to focus on dimension $r$ as opposed to dimension $\le r$, this is not strictly true as one can achieve a better list size $L$ for dimension $s\ll r$.  Similarly, the above lemma relating strong and weak designs seems to suggest that qualitatively (up to polynomial factors) these notions are the same.  However, as we will later (\autoref{sec:dim-exp_constr}), obtaining the appropriate (strong) list size simultaneously for all $s\le r$ will be crucial for our application to constant-degree dimension expanders.

\subsection{Seeded Lossless Rank Condensers}

Subspace designs ask that for any small subspace $V$ there is some $H_i\in\cH$ so that $H_i\cap V$ is \emph{small}.  Equivalently, the amount of dimension in $V$ that is outside $H_i$ is \emph{large} so that in some sense the dimension of $V$ is preserved.  This perspective is more naturally phrased in the language of \emph{(seeded) rank condensers}, as defined by Forbes, Saptharishi and Shpilka~\cite{ForbesSS14}.  The definition we use here is tuned to the equivalence with subspace designs, and we recover their definition as the lossless version of what we term here a \emph{lossy seeded rank condenser} (see \autoref{defn:lossy-seeded}).  We will discuss prior work and motivation for rank condensers that is less immediately relevant in \autoref{sec:lossless-seeded_constr}. We begin with the definition.

\begin{definition}\label{defn:lossless-seeded}
	Let $\F$ be a field and $n\ge r\ge 1$.  A collection of matrices $\cE\subseteq\F^{t\times n}$ is a \demph{weak (seeded) $(r,L)$-lossless rank condenser} if for all matrices $M\in\F^{n\times r}$ with $\rank M=r$,
	\[
		|\{E\st E\in\cE, \rank EM<\rank M\}|\le L
		\;.
	\]
	The collection $\cE$ is a \demph{strong (seeded) $(r,L)$-lossless rank condenser} if for all matrices $M\in\F^{n\times r}$ with $\rank M=r$,
	\[
		\sum_{E\in\cE} (\rank M-\rank EM)\le L
		\;.
	\]
	The collection $\cE$ is \demph{explicit} if given an index $i\in[|\cE|]$ the $i$-th matrix of $\cE$ can be constructed in $\poly(t,n,\log|\cE|)$ operations in $\F$.
\end{definition}

As we have many types of condensers in this paper (weak, strong, lossless, lossy, two-source, etc.) we will often just refer to them as ``condensers'' (perhaps with some relevant parameters such as ``$(r,\eps)$'') when the relevant adjectives are clear from context.

As it can only increase the quality of the condenser, one naturally considers the case when $\rank E=t$ for all $E\in\cE$.  However, we do not impose this restriction just as we do not impose the condition that subspaces in subspace designs all have the same dimension.  In fact, by the equivalence of subspace designs and lossless rank condensers (\autoref{res:subspace-design-equal-lossless}) one can see that these two restrictions are equivalent.

We briefly remark that as all of the pseudorandom objects we consider in this work are linear (or in the case of two-source condensers, bilinear) we will often freely pass between subspaces $V\subseteq\F^n$ of dimension $r$ and matrices $M\in\F^{n\times r}$ of rank $r$, using that we can choose a basis for $V$ so that $\cspn M=V$.  As such, we will often treat a matrix $M\in\F^{n\times r}$ as a list of $r$ vectors in $\F^n$. 

We now note that subspace designs are equivalent to lossless rank condensers.

\begin{propositionwp*}[\autoref{res:subspace-design-equal-lossless}]
	Let $\F$ be a field and $n\ge r\ge 1$.  Let $\cH=\{H_i\}_{i\in[M]}$ be a collection of subspaces $H_i\subseteq\F^n$ and let $\cE=\{E_i\}_{i\in[M]}\subseteq\F^{t\times n}$ be a collection of matrices, where we have that $\rspn E_i=(H_i)^\perp$ for $i\in[M]$. Then $\cH$ is a weak/strong $(r,L)$-subspace design iff $\cE$ is a weak/strong $(r,L)$-lossless rank condenser.
\end{propositionwp*}

While the above proposition is quite simple, it offers a unifying perspective of these different objects which was key to obtaining further results.

\subsection{Seeded Lossy Rank Condensers}

While the above seeded lossless rank condensers already have applications to list-decodable codes, rank condensers were defined in Forbes, Saptharishi and Shpilka~\cite{ForbesSS14} for quite different reasons.  We now give a definition closer to their motivation.

\begin{definition}\label{defn:lossy-seeded}
	Let $\F$ be a field and $n\ge r\ge 1$ and $\eps\ge 0$.  A collection of matrices $\cE\subseteq\F^{t\times n}$ is a \demph{(seeded) $(r,\eps)$-lossy rank condenser} if for all matrices $M\in\F^{n\times r}$ with $\rank M=r$, 
	\[
		\rank EM\ge (1-\eps)\rank M
		\;,
	\]
	for some $E\in\cE$. The collection $\cE$ is a \demph{(seeded) $(\le r,\eps)$-lossy rank condenser} if it a $(s,\eps)$-lossy condenser for all $1\le s\le r$.

	The collection $\cE$ is \demph{explicit} if given an index $i\in[|\cE|]$ the $i$-th matrix of $\cE$ can be constructed in $\poly(t,n,\log|\cE|)$ operations in $\F$.
\end{definition}

This notion is a natural linear-algebraic analogue of condensers for \emph{min-entropy} from the realm of Boolean pseudorandomness.  One contrast is that we do not require that \emph{most} $E\in\cE$ have the desired condensing property as this does not seem important for our applications, although we note that one can also obtain this stronger requirement with our methods.

In is worthwhile to contrast this object with subspace designs or lossless rank condensers. The goal of subspace designs was (due to connections with list-decodable codes) to construct a \emph{large} design while less focus was on the exact list-size bound.  Here, we have the somewhat different goal of obtaining a \emph{small} collection of matrices, which is akin to obtaining a very small list size in a subspace design. The focus on the collection being small is from the use of such condensers in derandomization, as we will need to enumerate over each matrix in the collection.

In particular, the notion of a $(r,0)$-lossy rank condenser is of interest because it is \emph{lossless}, which is important for many applications.  In particular, this notion was previously defined as a ``\emph{rank condenser (hitting set)}'' in the work of Forbes, Saptharishi and Shpilka~\cite{ForbesSS14}, but the construction and usage of these objects predates them\footnote{We note that the works we highlight are not necessarily the first or last in their respective lines of research, and rather we only highlight those that (to the best of our knowledge) had results concerning lossless rank condensers.}.  In particular, Gabizon and Raz~\cite{GabizonRaz08} constructed a $(r,0)$-condenser with size $nr^2$, and they used this to construct affine extractors over large fields. Karnin and Shpilka~\cite{KarninShpilka11} named the construction of Gabizon and Raz~\cite{GabizonRaz08} to be ``rank preserving subspaces'' and used this construction to make a \emph{polynomial identity testing}\footnote{The \emph{polynomial identity testing problem} is when given a algebraic circuit $C$ (perhaps from a restricted class of circuits) to \emph{deterministically} decide whether the circuit $C$ computes the identically zero polynomial. The \emph{black box} version is where we only allow access to $C$ by evaluating the polynomial it computes. See Shpilka and Yehudayoff~\cite{ShpilkaYehudayoff10} for more on this problem.} algorithm of Dvir and Shpilka~\cite{DvirShpilka07} work in the \emph{black box} model.   Forbes and Shpilka~\cite{ForbesShpilka12} later gave an improved construction of a rank condenser with only $nr$ size, and showed how they can be used to make another polynomial identity testing algorithm of Raz and Shpilka~\cite{RazShpilka05} work in the black-box model. Forbes, Saptharishi and Shpilka~\cite{ForbesSS14}, building on the work of Agrawal, Saha, and Saxena~\cite{AgrawalSS13}, analyzed ``multivariate'' lossless rank condensers as they arose naturally in a polynomial identity testing algorithm.

Beyond applications to polynomial identity testing, Lokshtanov, Misra, Panolan and Saurabh~\cite{LokshtanovMPS14} used these condensers to derandomize a fixed-parameter-tractable algorithm of Marx~\cite{Marx09} for $\ell$-matroid intersection.  Cheung, Kwok and Lau~\cite{CheungKL13} rediscovered the rank condenser of Gabizon and Raz~\cite{GabizonRaz08} and (among other things) used this to give faster randomized algorithms for exact linear algebra. Forbes, Saptharishi and Shpilka~\cite{ForbesSS14} showed a generic recipe to construct such rank condensers from \emph{any} error-correcting code (over large fields). Given these applications and connections present in $(r,\eps)$-lossy rank condensers for $\eps=0$, we expect the $\eps>0$ version will similarly have many applications.

We now quote the parameters given by the probabilistic method.

\begin{propositionwp*}[Informal version of \autoref{res:lossy-seeded_prob-method=} and \autoref{res:lossy-seeded_prob-method}]
	Let $\F_q$ be a finite field. Let $n\ge r\ge 1$, $\eps\ge 0$ and $t>(1-\eps)r$.   Then there is a collection $\cE$ of $k$ matrices $\cE\subseteq\F_q^{t\times n}$ that is a $(r,\eps)$-lossy rank condenser whenever
	\begin{equation}\label{eq:temp-lossy-cond}
		k\ge\frac{rn+o_q(1)}{(t-(1-\eps)r)(\floor{\eps r}+1)-o_q(1)}
		\;.
	\end{equation}
	For $\eps>0$, there is a collection $\cE$ of size $k$ that is a $(\le r,\eps)$-lossy rank condenser whenever
	\[
		k\ge\frac{n+o_q(1)}{\eps (t-(1-\eps)r)-o_q(1)}
		\;.
		\qedhere
	\]
\end{propositionwp*}

Thus we can make the output size $t$ of the condenser to be almost equal to the guaranteed dimension bound of $(1-\eps)r$.  Further, we see that there is essentially no penalty in (existentially) insisting for a $(\le r,\eps)$-condenser over a $(r,\eps)$-condenser.  However, we show in \autoref{res:dim-le-r_dim-eq-r} that the notion of $(\le r,\eps)$-condenser is provably stronger.

\subsection{Our Results}

We now turn to our constructions of condensers.  We begin with the following construction, which is the rank condenser of Forbes and Shpilka~\cite{ForbesShpilka12} and was named the \emph{folded Wronskian} by Guruswami-Kopparty~\cite{GuruswamiKopparty13}.

\begin{construction}[Folded Wronskian]\label{constr:folded-wronskian}
	Let $\F$ be a field.  Let $\omega\in\F$ be an element of multiplicative order $\ge n$. Define $\Wr_{t,\omega}(x)\in\F[x]^{\zr{t}\times \zr{n}}$ by $(\Wr_{t,\omega}(x))_{i,j}\eqdef (\omega^ix)^j$.

	That is, identifying $\F^{\zr{n}}$ with the degree $<n$ polynomials $\F[x]^{<n}$, we see that $\Wr_{t,\omega}(x):\F[x]^{<n}\to\F[x]^t$ defined by
	\[
		f(x)\mapsto (f(x),f(\omega x),\cdots, f(\omega^{t-1} x))
		\;.
		\qedhere
	\]
\end{construction}

That is, we define $\zr{n}\eqdef\{0,\ldots,n-1\}$ so that in the above $i$ and $j$ are indexed from zero. When the value of $\omega$ is clear from context we will just write ``$\Wr_t$''.  Note that the fact that $\omega$ has large multiplicative order means that we require a large field, in particular that $|\F|>n$.

The key result that forms the starting point for our constructions is the following analysis of the folded Wronskian by Guruswami and Kopparty~\cite{GuruswamiKopparty13}.  While their analysis was originally in the context of subspace designs, we state their result here in the language of lossless rank condensers as it is more natural in our context.

\begin{restatable}[Guruswami-Kopparty~\cite{GuruswamiKopparty13}]{theoremwp}{foldedwronskian}\label{res:folded-wronskian-GK13}
	Assume the setup of \autoref{constr:folded-wronskian} where we take $t\ge r\ge 1$. Let $S\subseteq \{(\omega^\ell)^j\st j\ge 0\}$ where $\ell\ge t-r+1$.  Then $\{\Wr_t(\alpha)\st \alpha\in S\}\subseteq\F^{t\times n}$ is a strong $(r,\frac{r(n-r)}{t-r+1})$-lossless rank condenser.
\end{restatable}

We note here that the above parameters are slightly stronger than what Guruswami and Kopparty~\cite{GuruswamiKopparty13} obtain, as they only obtain a list bound of $\frac{r(n-1)}{t-r+1}$. This improved bound follows by using some of the analysis from Forbes, Saptharishi and Shpilka~\cite{ForbesSS14} as explained in \autoref{sec:lossless-seeded_constr}. Note that this construction essentially matches the non-constructive bound \eqref{eq:temp-lossy-cond} when $\eps = 0$.

The above analysis indicates that for a matrix $M\in\F^{n\times r}$ of rank $r$ that the total rank loss over all maps in $\cE$ is at most $\frac{r(n-r)}{t-r+1}$.  Thus, by an averaging argument, at most $\nicefrac{1}{k}\cdot \frac{r(n-r)}{t-r+1}$ such maps can have a rank loss of $\ge k$.  This observation thus shows that the above construction is not just a \emph{lossless} rank condenser but also a \emph{lossy} condenser (with different parameters).

\begin{corollarywp*}[\autoref{res:gk13-lossy-condenser_explicit}]
	Let $\F$ be a field. Let $n,t\ge r\ge 1$ and $\eps>0$, where $\omega\in\F$ is an element of multiplicative order $\ge\poly(n)$. Define $\cE\eqdef\{\Wr_{t,\omega}((\omega^t)^j)\st 0\le j< \frac{n}{\eps(t-r+1)}\}$, that is, the folded Wronskian evaluated at $\frac{n}{\eps(t-r+1)}$ distinct powers of $\omega^t$. Then $\cE$ is an explicit $(\le r,\eps)$-lossy rank condenser.
\end{corollarywp*}

To motivate our below application to dimension expanders, suppose that $r=\nicefrac{n}{3}$ and $t=\nicefrac{n}{2}$ and some $\eps>0$. This says then that we construct a rank condenser that maps $\F^n$ to $\F^{\nicefrac{n}{2}}$ that maps rank $\nicefrac{n}{3}$ subspaces to rank $(1-\eps)\nicefrac{n}{3}$ subspaces.  Further, this condenser is a collection of at most
\[
	\frac{n}{\eps(\nicefrac{n}{2}-\nicefrac{n}{3})}=\nicefrac{6}{\eps}
\]
maps such that one map from the collection always preserves the desired rank.  To obtain these parameters, it is key to the analysis that we have a \emph{strong} lossless condenser and that it obtains the (near-optimal) bound given by Guruswami and Kopparty~\cite{GuruswamiKopparty13}. Note that these condensing parameters are very much similar to the min-entropy condenser of Raz~\cite{Raz05}, who uses a constant number of random bits to condense a source with constant-rate min-entropy.  

\section{Dimension Expanders}\label{sec:dim-exp}

We now turn to our main object of interest, \emph{dimension expanders}. Dimension expanders were defined by Barak, Impagliazzo, Shpilka and Wigderson~\cite{BISW04} in an attempt to translate challenges in the explicit construction of objects in Boolean pseudorandomness into the regime of linear algebra.  Indeed, in combinatorics there is a well-established analogy between subsets of $[n]$ and subspaces of vector spaces over finite fields.  In the context of pseudorandomness, we can then translate questions that manipulate the \emph{size} of subsets $S\subseteq\bits^n$ (or more generally, the min-entropy of distributions over $\bits^n$) into questions about manipulating the \emph{dimension} of subspaces $V\subseteq\F^n$. While these regimes seem different, it is conceivable that such linear algebraic constructions could yield new constructions in Boolean pseudorandomness (such as how the inner-product function is a two-source extractor).  Indeed, as in the work of Guruswami and Wang~\cite{GuruswamiWang13}, this idea has borne fruit (if in a perhaps unexpected way) by showing how linear-algebraic pseudorandom objects can improve list-decodable codes. We now define dimension expanders.

\begin{definition}\label{defn:dim-exp}
	Let $\F$ be a field, $n\ge 1$, $\eps>0$ and $\alpha\in \R$ with $\alpha\ge 1$. A collection of matrices $\cA=\{A_1,\ldots,A_d\}\subseteq\F^{n\times n}$ is a \demph{$(\eps,\alpha)$-dimension expander of degree $d$} if for all subspaces $V\subseteq\F^n$ of dimension $\le \eps n$ that
	\[
		\dim \sum_{i=1}^d A_i(V)
		=\dim\spn\{A_i(V)\}_{i=1}^d
		\ge \alpha \dim V
		\;.
	\]
	The collection $\cA$ is \demph{explicit} if given an index $i\in[|\cA|]$ the $i$-th matrix in $\cA$ can be constructed in $\poly(n,\log|\cA|)$ operations in $\F$.
\end{definition}

We remark that in the above definition one can generally assume that all of the maps $A_i$ are of full-rank, as that can only increase $\dim \sum_{i=1}^d A_i(V)$.  Similarly, one can assume that $A_1$ equals the identity matrix $\Id_n$ as we can use the transform $A_i\mapsto A_1^{-1}A_i$ as again this does not affect the size of the outputted dimension.  While these assumptions are thus without loss of generality, we will not impose them.

In general we will be most interested in $(\Omega(1),1+\Omega(1))$-dimension expanders of constant degree, which we shall thus call ``dimension expanders'' without any quantification.  This parameter regime is of interest because it matches that of the probabilistic method, which we quote the results of below.

\begin{propositionwp*}[Informal version of \autoref{res:dim-exp_prob-method}]
	Let $\F_q$ be a finite field, $n\ge 1$, $\eps>0$ and $\alpha\in \R$ with $\alpha\ge 1$. Then there exist a collection matrices $\cA=\{A_1,\ldots,A_d\}\subseteq\F^{n\times n}$ which is a $(\eps,\alpha)$-dimension expander of degree $d$ whenever
	\[
		d\ge \alpha+\frac{1}{1-\alpha\eps}+o_q(1)
		\;.
		\qedhere
	\]
\end{propositionwp*}

Put into more concrete terms, we see that one can existentially obtain $(\nicefrac{1}{2d},d-\O(1))$-dimension expansion with degree $d$.  That we have an expansion of $(1-\eps)d$ in a degree $d$ expander is akin to \emph{lossless (vertex) expanders} which have a similar degree/expansion relation, and these expanders have applications beyond those of normal expanders (see Capalbo, Reingold, Vadhan and Wigderson~\cite{CapalboRVW02} and references therein).  While previous work focused on obtaining constant-degree dimension expanders, our work raises the questions of obtaining \emph{lossless} dimension expanders so that we match the above bound.  Our work, as discussed below, lends itself to being particularly quantitative with regards to the size and parameters of the construction.  However, we do not obtain lossless dimension expanders, and to the best of our knowledge, neither do the other previous constructions of dimension expanders discussed below.

While we discuss prior work in depth in \autoref{sec:dim-exp_constr}, we briefly summarize the state of art in dimension expanders in the following theorems.  

\begin{theoremwp*}[Lubotzky and Zelmanov~\cite{LubotzkyZelmanov08} and Harrow~\cite{Harrow08}]
	Let $\F$ be a field of characteristic zero and $n\ge 1$. There exists an explicit $\O(1)$-sized collection $\cA\subseteq\F^{n\times n}$ such that $\cA$ is a $(\nicefrac{1}{2},1+\Omega(1))$-dimension expander over $\F^n$.
\end{theoremwp*}

This construction requires characteristic zero as it uses a notion of distance that lacks a good definition in finite characteristic.

\begin{theoremwp*}[Bourgain and Yehudayoff~\cite{BourgainYehudayoff13}]
	Let $n\ge 1$. 
	There exists an explicit $\O(1)$-sized collection $\cA\subseteq\bits^{n\times n}$ such that $\cA$ is a $(\nicefrac{1}{2},1+\Omega(1))$-dimension expander over $\F^n$, over every field $\F$.
\end{theoremwp*}

Note that the above construction is only a function of $n$, and not of the field, so that this \emph{single} construction is a dimension expander over \emph{all} fields.

As explained in \autoref{sec:dim-exp_constr}, both of the above constructions in some way attempt to extend existing ideas about expander graphs into the world of dimension expanders.  The first replicates the representation theory approach to constructing expanding Cayley graphs, and the second shows how bipartite expanders (with the strong requirement of \emph{monotonicity}) extend to also be dimension expanders.  

\smallskip \noindent {\bf Our Work.} In our work we take a different approach to constructing dimension expanders that treats such expanders as part of an emerging theme of \emph{linear-algebraic} pseudorandomness as seen by recent linear-algebraic approaches to list-decoding~\cite{Guruswami11,GuruswamiXing12,GuruswamiWang13,GuruswamiXing13,GuruswamiWang14} and linear-algebraic derandomization of subclasses of polynomial identity testing~\cite{KarninShpilka11,ForbesShpilka12}.  The first consequence of this perspective is that we work in fields that are at least polynomially large as this is the setting of Reed-Solomon codes.  To obtain dimension expanders over smaller fields, a natural solution within this theory is to use ``code concatenation'' ideas from coding theory.  Unfortunately the idea of code concatenation is somewhat subtle in our setting and so only supplies a concatenation (based on converting Reed-Solomon codes to BCH codes) that incurs a logarithmic loss in the parameters. The second consequence is that we build our dimension expanders out of the existing linear-algebraic pseudorandom objects that have emerged from prior work.  That is, just how in Boolean pseudorandomness the notions of expanders, extractors and list-decodable codes are all related (see for example Vadhan~\cite{Vadhan12}), we leverage such connections to construct our expanders from the above mentioned rank condensers.

We now explain our construction, which while ultimately was motivated by the connections between two-source rank condensers and dimension expanders (\autoref{res:two-src_dim-exp}), can be explained in a self-contained manner.  The first observation is that one can easily obtain ``$(1,d)$-expanders'' of degree $d\in\N$ if one is willing to allow the ambient space to grow.  That is, consider the tensor product $\F^n\otimes\F^d=\F^{nd}$.  By properties of the tensor product, for $V\subseteq\F^n$ of rank $r\le n$ we know that $V\otimes \F^d$ is of rank $rd$ in $\F^{nd}$.  Further, $V\otimes\F^d$ can be seen as the image of $d$ maps $T_i:\F^n\to\F^{nd}$ where the $i$-th map places the space $\F^n$ into the ``$i$-th block'' of $(\F^n)^d=\F^{nd}$.  In analogy to bipartite expander graphs, this is akin to giving each left vertex its own ``private neighborhood'' of right vertices into which it expands. 

While trivial, the above step now allows us to convert a question of \emph{expansion} to a question of \emph{condensing}.  That is, tensoring achieves expansion only because the output of the maps are larger than the input, while the non-trivial aspect of dimension expanders is to expand while keeping the output size the \emph{same}.  However, tensoring \emph{has} expanded dimension and thus we can now focus on reducing the output size. Specifically, suppose that we consider $V\subseteq\F^n$ of rank $r=\nicefrac{n}{2d}$.  Then its image under the above tensoring is $W\eqdef\sum_i T_i(V)$ of dimension $\nicefrac{n}{2}$.  This subspace $W$ lies in an $nd$-dimensional space and we wish return it to an $n$-dimensional space while not loosing too much in the dimension. However, this last problem is exactly the question of \emph{lossy rank condensing}.  As shown above after the informal version of \autoref{res:gk13-lossy-condenser_explicit}, we can condense such constant-rate dimension in a lossy way using a \emph{constant} number of maps.  In this example, we can condense $W$ to $\F^n$ using $\frac{dn}{\eps(n-\nicefrac{n}{2})}=\frac{2d}{\eps}$ maps, at least one of which produces a $(1-\eps)\nicefrac{n}{2}$ dimensional space.  Thus, this expands $V\subseteq\F^n$ of rank $\frac{n}{2d}$ to be of dimension $(1-\eps)\nicefrac{n}{2}$ within $\F^n$, all while using $d\cdot \frac{2d}{\eps}=\frac{2d^2}{\eps}$ maps (we multiply the number of maps due to the composition).  We summarize this composition in \autoref{fig:tensor-then-condense}.

\begin{figure}
	\centering
	\begin{tikzpicture}[minimum size=.3in,auto]
		\node [minimum width=.5in] (a) at (-2.0in,0) {$\F^n$};
		\node [minimum width=.5in] (b) at (0,0) {$\F^{nd}$};
		\node [minimum width=.5in] (c) at (2.0in,0) {$\F^{n}$};
		\draw [->] (a) to node[above] {tensoring} (b);
		\draw [->] (b) to node[above,minimum size=.4in] {\parbox{1.5in}{\centering $(\le \eps d n,\delta)$-lossy\\condenser}} (c);
		\draw [->] (a) to node[below] {degree $d$} (b);
		\draw [->] (b) to node[below] {degree $\frac{dn}{\delta(n-\eps dn)}$} (c);
		\node [minimum width=.8in] (aa) at (-2.0in,-.7in) {dim $\eps n$};
		\node [minimum width=.8in] (bb) at (0,-.7in) {dim $\eps dn$};
		\node [minimum width=1.2in] (cc) at (2.0in,-.7in) {dim $(1-\delta)\eps d n$};
		\draw [->] (aa) to (bb);
		\draw [->] (bb) to (cc);
	\end{tikzpicture}
	\caption{Constructing dimension expanders from tensoring and lossy rank condensers.}
	\label{fig:tensor-then-condense}
\end{figure}
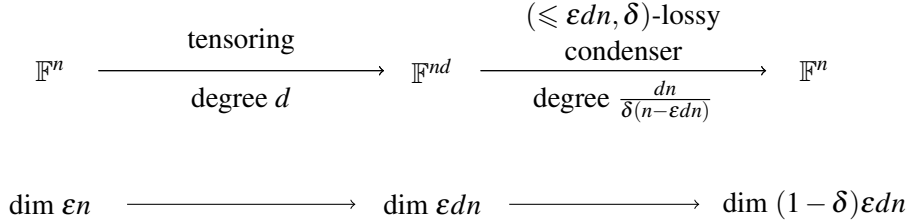

We note that the above discussion has only discussed constant-rate rank, that is, subspaces of $\F^n$ with rank $\Omega(n)$.  Dimension expanders however are required to expand \emph{all} small subspaces.  Our construction also handles this case as the lossy rank condensers we use will preserve a $(1-\delta)$ fraction of the input rank, as long as that rank is small enough. In the above sketch there is also the technicality that we must tensor with $\F^d$ with $d$ being \emph{integral}, which restricts $d\ge 2$ as $d=1$ does not yield expansion.  With this construction alone one would only obtain expansion in $\F^n$ for rank $<\nicefrac{n}{d}\le\nicefrac{n}{2}$, but we manage to sidestep this restriction by a simple truncation argument.  Putting the above pieces together we obtain the following theorem.

\begin{theoremwp*}[Main Theorem, Informal version of \autoref{res:tensor-then-condense_instantiate_gamma0} and \autoref{res:tensor-then-condense_instantiate_gamma-pos}]
	Let $n,d\ge 1$ and let $0<\eps\le\eta<1$ be constants. Let $\F$ be a field with $|\F|\ge\poly(n)$. 
	There is an explicit $(\eps,\nicefrac{\eta}{\eps})$-dimension expander in $\F^n$ of degree $\Theta\left(\frac{1}{\eps^2(1-\eta)^2}\right)$.
	If $\eps<\nicefrac{1}{d}$ then there is an explicit $(\eps,(1-\delta)d)$-dimension expander in $\F^n$ with degree $\frac{d^2}{\delta(1-\eps d)}$. 
\end{theoremwp*}

These expanders yield an expansion of $\alpha$ with degree $\approx \alpha^2$, and thus are not lossless.  In particular, the existential bound of \autoref{res:dim-exp_prob-method} shows that there are $(\eps,\nicefrac{\eta}{\eps})$-dimension expanders with degree $\approx \nicefrac{1}{\eps}+\frac{1}{1-\eta}$. In remains an interesting challenge to obtain such lossless dimension expanders.  In particular, we note that we get ``all of the dimension'' from the tensoring step using only \emph{one} map from the condenser. This occurs despite the fact that \emph{most} maps in the condenser preserve all of the dimension (assuming we double the seed length). It seems natural to hope that an integrated analysis of the tensoring and condensing stages would show that the construction has a better expansion than what we obtain.

Over small fields our results are comparatively weaker as we simulate a larger field within the small field (as how one transforms Reed-Solomon codes to BCH codes), so that we pay various logarithmic penalties.

\begin{corollarywp*}[Informal version of \autoref{res:tensor-then-condense_instantiate_small-field}]
	Let $\F_q$ be a finite field. Let $n,d\ge 1$. Then there are explicit $\left(\Theta\left(\frac{1}{d\log_q dn}\right),\Theta(d)\right)$-dimension expanders in $\F_q^n$ of degree $\Theta(d^2\log_q dn)$.
\end{corollarywp*}

\section{Two-Source Rank Condensers}\label{sec:two-src}

In the context of Boolean pseudorandomness, it is well known (see for example Vadhan~\cite{Vadhan12}) that strong min-entropy seeded extractors (extractors that output the entropy of the source \emph{plus} the entropy of the seed) are equivalent to a form of vertex expansion.  Such extractors are a special case of (seedless) two-source min-entropy extractors where one of the sources is very small and of full entropy. Thus, as a generalization of the dimension expanders we have already defined, we can thus define the notion of a \emph{(seedless) two-source rank condenser}.  While it is often most natural to consider the two sources to be of equal dimension, to highlight the connection to dimension expanders (\autoref{res:two-src_dim-exp}) we consider sources with unbalanced dimension. 

\begin{definition}\label{defn:two-source}
	Let $\F$ be a field and $n\ge r\ge 1$ and $m\ge s\ge 1$. A function $f:\F^n\times\F^m\to\F^t$ is a \demph{(seedless) $(r,s,\eps)$-two-source rank condenser} if for all sets $A\subseteq\F^n$ and $B\subseteq\F^m$ with $\rank A=r$ and $\rank B=s$,
	\[
		\rank f(A\times B)=\rank \{ f(\vv,\vw)\}_{\vv\in A,\vw\in B} \ge (1-\eps)\rank A\cdot \rank B
		\;.
	\]
	The function $f$ is a $(\le r,s,\eps)$-condenser if it is a $(r',s,\eps)$-condenser for all $1\le r'\le r$, and $(\le r,\le s,\eps)$-condensers are defined similarly. If $\eps=0$ we say the rank condenser is \demph{lossless} and it is otherwise \demph{lossy}. The function $f$ is \demph{bilinear} if $f(\vv,\vw)=(\vv^\T E_i\vw)_{i=1}^t$ for $E_i\in\F^{n\times m}$. The function $f$ is \demph{explicit} if it can be evaluated in $\poly(n,m,t)$ steps.
\end{definition}

While this definition is naturally motivated as a generalization of dimension expanders, we originally were motivated to study these objects due to potential applications for constructing \emph{subspace evasive sets}, as we describe in \autoref{sec:subspace-evasive}.

Note that in general we allow the function $f$ to be arbitrary, but in this work we will restrict ourselves to bilinear functions $f$ as they are the most natural.  In this case, as discussed after \autoref{defn:lossless-seeded}, we see that the function $f$ acts on \emph{subspaces} so that we ask that for subspaces $V\subseteq\F^n$ and $W\subseteq\F^m$ that $\dim f(V,W)\ge (1-\eps)\dim V\cdot \dim W$.  In this way, $f$ can be thought of as a \emph{derandomized tensor product}.

We now quote the parameters as given by the probabilistic method.

\begin{propositionwp*}[Informal version of \autoref{res:two-src_prob-method} and \autoref{res:two-src_prob-method<}]
	Let $\F_q$ be a finite field. Let $n\ge r\ge 1$ and $m\ge s\ge 1$ and $\eps\ge0$. Then there exists a function $f:\F^n\times\F^m\to\F^t$ which is a bilinear $(r,s,\eps)$-two-source rank condenser, assuming that
	\[
			t
			\ge
			\frac{n}{\eps s}+\frac{m}{\eps r}+(1-\eps) rs+o_q(1)
			\;.
	\]
	for $\eps>0$. Further, there exists a $f$ which is a  $(\le r,s,\eps)$-condenser assuming that
	\[
		t
		\ge
		\frac{n}{\eps s}+\frac{m}{\eps }+(1-\eps) rs+o_q(1)
		\;.
	\]
	
	If $\eps=0$, then there exists an $f$ which is a $(r,s,0)$-condenser assuming that
	\[
		t\ge rn+sm+rs+o_q(1)
		\;.
		\qedhere
	\]
\end{propositionwp*}

In particular, in the balanced case of $n=m$ and $r=s$ this shows that any $t\ge \frac{2n}{\eps r}+(1-\eps) r^2+o_q(1)$ suffices. Note that unlike the single-source setting, there is a large penalty for condensing all small enough sources.  Thus, the above gives $(r,r,\nicefrac{1}{2})$-condensers with output $\approx \frac{n}{r}+r^2$ but to obtain a $(\le r,r,\nicefrac{1}{2})$-condenser the resulting output size is $\approx n+r^2$ (and \autoref{res:two-source_output-lb} shows that a linear dependence on $n$ is needed in this case).

Note that in our definitions of seeded rank condensers there was no analogue of \emph{strong} min-entropy extractors, which are extractors that also recover the entropy of the seed in addition to the entropy of the source.  That is, in our setting, there is no ``rank of the seed'' to recover as the seed is simply an index into the collection $\cE$.  The notion of a two-source rank condenser in some sense allows the second source to be a ``seed'' in that we can associate elements of $\cE$ with elements in a basis for $\F^m$.  However, we do not pursue this analogy further as two-source rank condensers meeting the probabilistic method (\autoref{res:two-src_prob-method<}) do not seem to yield good lossy rank condensers in all regimes as two-source condensers can require an output size which is linear in the input size (\autoref{res:two-source_output-lb}).

However, the connection between two-source extractors and expanders does hold tightly for the notion of rank, as we show. Note that for this connection it suffices to have condensers that work when one of the two sources has full rank.

\begin{propositionwp*}[Informal version of \autoref{res:two-src_dim-exp}]
	Let $\F_q$ be a finite field. For large $n$ and all other parameters constant, constructions of bilinear $(\le \delta n,m,\eps)$-two-source rank condensers $f:\F^n\times\F^m\to\F^t$ that meet the probabilistic method bound of \autoref{res:two-src_prob-method<} yield constructions of $(\delta,\alpha)$-dimension expanders in $\F^n$ meeting the probabilistic method bound of \autoref{res:dim-exp_prob-method}.
\end{propositionwp*}

We also give constructions of two-source condensers using seeded rank condensers.  That is, for two sources we use a seeded rank condenser to condense each source and use a union bound to show that the seed-length only doubles. We then enumerate over seeds and for each seed we then tensor the two condensed sources together.  While this approach seems wasteful, we show that it yields \emph{optimal} lossless two-source rank condensers by appropriate pruning. In particular, we observe that this is the same construction as given by Forbes and Shpilka~\cite{ForbesShpilka12} for an object known as a \emph{rank-metric code}.  We push this observation further to see (in \autoref{sec:rank-metric}) that bilinear lossless two-source rank condensers are \emph{equivalent} to rank-metric codes.  Using this connection, we obtain optimal such condensers over \emph{any} field using known constructions of rank-metric codes.

\begin{theoremwp*}[Informal version of \autoref{res:two-src_constr} and \autoref{res:two-src_rs-r}]
	Let $\F$ be a field and $n\ge r\ge 1$ and $m\ge s\ge 1$. Then there is an explicit $f:\F^n\times\F^m\to\F^t$ which is a $(r,s,0)$ with $t\le\O(\max\{r,s\}(n+m))$.
\end{theoremwp*}

We then turn to constructions of \emph{lossy} two-source condensers, where our results are considerably weaker.  However, we are able to give near-optimal results for \emph{constant} $r$ by using a brute force ``inner condenser'' and using our condense-then-tensor results as an ``outer condenser''.

\begin{propositionwp*}[Informal version of \autoref{res:two-src_lossy-outer}]
	Let $\F$ be a field and $n\ge r\ge 1$, where $\F$ is polynomially large and $r\le\O(1)$.  Then there is an explicit bilinear $(r,r,1-(1-\eps)^3)$-two source rank condenser $f:\F^n\times\F^n\to\F^t$ with $t\le\O(\nicefrac{n}{\eps^2 r})$.
\end{propositionwp*}

\section{Notation}

We briefly summarize some notation we use. We will apply functions not just to inputs, but to sets of inputs.  Thus, if $f:S\to T$ and $U\subseteq S$, then $f(U)\eqdef \{f(x)\}_{x\in U}$.  We denote $[n]\eqdef\{1,\ldots,n\}$ and $\zr{n}=\{0,\ldots,n-1\}$.  We sometimes use $\zr{n}$ to index matrices from $0$, so that $M\in\F^{\zr{n}\times m}$ has its rows indexed from zero but its columns indexed from 1.  The $i$-th row of a matrix $M$ will be denoted by $M_{i,\bullet}$ and the $j$-th column denoted $M_{\bullet,j}$. The transpose of a matrix $M$ will be denoted by $M^\T$.

\parskip=1ex
\section{Constructions of Subspace Designs and Rank Condensers}\label{sec:lossless-seeded_constr}

In this section we relate subspace designs to lossless rank condensers, and use this relation to construct lossless and lossy seeded rank condensers.  We begin with this relation, as given by taking dual spaces.  We then construct lossless rank condensers using the folded Wronskian from \autoref{constr:folded-wronskian} and the analysis of Guruswami-Kopparty~\cite{GuruswamiKopparty13} quoted in \autoref{res:folded-wronskian-GK13} (which we slightly improve using the analysis of Forbes, Saptharishi and Shpilka~\cite{ForbesSS14}).  We then turn to lossy condensers, first noting that condensing ``dimension $\le r$'' and ``dimension $r$'' is more subtle in this regime.  We then show how to convert any strong lossless condenser to a lossy condenser and in particular do so for our explicit lossless condenser. As a result we obtain a lossy condenser nearly matching the probabilistic method.

\subsection{Lossless Rank Condensers}

We first show that subspace designs and lossless seeded rank condensers are the \emph{same} object.

\begin{proposition}\label{res:subspace-design-equal-lossless}
	Let $\F$ be a field and $n\ge r\ge 1$.  Let $\cH=\{H_i\}_{i\in[M]}$ be a collection of subspaces $H_i\subseteq\F^n$ and let $\cE=\{E_i\}_{i\in[M]}\subseteq\F^{t\times n}$ be a collection of matrices, where we have that $\rspn E_i=(H_i)^\perp$ for $i\in[M]$. Then $\cH$ is a weak/strong $(r,L)$-subspace design iff $\cE$ is a weak/strong $(r,L)$-lossless rank condenser.
\end{proposition}
\begin{proof}
	There is a natural surjection from matrices $M\in\F^{n\times r}$ with rank $r$ to $r$-dimensional subspaces $V$ given by $M\mapsto \cspn M$. We now show that rank condensers act on matrices in the same way that subspace designs act on subspaces. Each matrix $E_i$ naturally defines a map $\varphi_i:\F^n\to\F^t$ where by definition $\ker \varphi_i=(\rspn E_i)^\perp=H_i$.  Basic linear algebra shows that $\rank M=\dim V=\dim \varphi_i(V) + \dim(\ker \varphi_i\cap V)$. Using that $\varphi_i(V)=\rank E_iM$, we have that $\rank M-\rank E_i M=\dim(H_i\cap V)$.  
	
	Thus, summing over $i$ we see that for $M$ with $V=\cspn M$,
	\[
		\sum_i (\rank M-\rank E_i M)= \sum_i \dim(H_i\cap V)
		\;.
	\]
	Thus, the left summation is bounded by $L$ for all $M$ iff the right summation is bounded by $L$ for all $V$, so that $\cE$ is $(r,L)$-strong iff $\cH$ is $(r,L)$-strong.

	Observing that $\rank E_iM<\rank M$ iff $\dim(H_i\cap V)>0$, the same conclusion holds for $(r,L)$-weak designs and condensers.
\end{proof}

Given the above equivalence, we now phrase the construction of Guruswami and Kopparty~\cite{GuruswamiKopparty13} of subspace designs in the language of lossless condensers.  To do so, we first need the following result of Forbes, Saptharishi and Shpilka~\cite{ForbesSS14} which gives an analysis of the folded Wronskian.

\begin{proposition}[{Implicit in Forbes, Saptharishi and Shpilka~\cite{ForbesSS14}, given explicitly by Forbes~\cite[Theorem 5.4.3]{Forbes14}}]\label{res:folded-wronskian-nonzero}
	Assume the setup of \autoref{constr:folded-wronskian}.  Let $M\in\F^{n\times r}$ be of rank $r$.  Then the $r\times r$ $\Wr_{r,\omega}(x)M\in\F[x]^{r\times r}$ has determinant $\det \Wr_{r,\omega}(x)M\in\F[x]$ which is a non-zero polynomial such that after dividing out powers of $x$ has degree $\le r(n-r)$.
\end{proposition}

There are roughly three components to the above result: the non-zeroness of the polynomial $\det \Wr_{r,\omega}(x)M$, the degree bound on $\det \Wr_{r,\omega}(x)M$, and the multiplicative order lower bound of the distinguished element $\omega$  needed for these facts.

The non-zeroness of $\det \Wr_{r,\omega}(x)M$ has been explicitly established by various works. Forbes and Shpilka~\cite{ForbesShpilka12} gave a proof based on the Cauchy-Binet formula, Guruswami and Wang~\cite{GuruswamiWang13} via analyzing a certain linear system, Guruswami and Kopparty~\cite{GuruswamiKopparty13} via an analysis akin to that of the folded Reed-Solomon codes of Guruswami and Rudra~\cite{GuruswamiRudra08}, and Lokshtanov, Misra, Panolan and Saurabh~\cite{LokshtanovMPS14} via a reduction to the $n=2$ case.  We remark that these proofs have some similarity to analyses of the (classical) Wronskian, as discussed in Guruswami-Kopparty~\cite{GuruswamiKopparty13} and Forbes-Saptharishi-Shpilka~\cite{ForbesSS14}, so that similar proofs can be extracted from that literature.

The above proofs of non-zeroness yield different degree upper bounds on $\det \Wr_{r,\omega}(x)M$ and the multiplicative order of $\omega$.  The obvious degree bound for $\det \Wr_{r,\omega}(x)M$ is $r(n-1)$ as this is an $r\times r$ determinant with entries who are polynomials of degree $<n$. Forbes and Shpilka~\cite{ForbesShpilka12} refined this bound to $nr-\binom{r+1}{2}$, and this is tight.  However, Forbes, Saptharishi and Shpilka~\cite{ForbesSS14} observed
\footnote{Forbes, Saptharishi and Shpilka~\cite{ForbesSS14} noted more generally that a similar construction follows from any error-correcting code, and that the code in the folded Wronskian is the dual Reed-Solomon code. However, the dual Reed-Solomon code seems special in that it gives rise to ``folding'' of the resulting construction, which is crucial for \autoref{res:folded-wronskian-GK13}.}
(and this is explicitly stated in Forbes~\cite{Forbes14}) that one can clear powers of $x$ to derive the degree bound of $r(n-r)$, and again this is tight. Regarding the multiplicative order of $\omega$, most of the above proofs show that order $\ge n$ is sufficient, and this is also tight.  

Given the above analysis, we now derive the weak lossless condenser of Forbes and Shpilka~\cite{ForbesShpilka12} (with an updated degree bound).

\begin{corollary}[Forbes and Shpilka~\cite{ForbesShpilka12}]\label{res:lossless-extractor}
	Assume the setup of \autoref{constr:folded-wronskian} where we take $t=r$. Let $S\subseteq \F\setminus\{0\}$.  Then $\{\Wr_r(\alpha)\st \alpha\in S\}\subseteq\F^{r\times n}$ is a weak $(r,r(n-r))$-lossless rank condenser.
\end{corollary}
\begin{proof}
	As $\rank \Wr_r(\alpha) M=\rank M=r$ iff $\det \Wr_r(\alpha)M\ne 0$, we see that by \autoref{res:folded-wronskian-nonzero} there are at most $r(n-r)$ such $\alpha\in\F\setminus\{0\}$ where $\rank \Wr_r(\alpha)M <\rank M$, giving the desired bound.
\end{proof}

Note that in the above condenser the output size $t$ is equal to the rank bound $r$ so that this might appropriately be called a rank \emph{extractor}\footnote{The paper of Dvir, Gabizon and Wigderson~\cite{DvirGW09} also defined a notion of a rank extractor, but that was for \emph{algebraic} rank (or transcendence degree) of polynomials which generalizes linear-algebraic rank to polynomials of higher degree.}.  In some applications this equality is important,
and indeed for our construction of dimension expanders we would ideally have lossy rank extractors (which we do not achieve, as discussed after \autoref{res:gk13-lossy-condenser_explicit}).

While the above degree and order analysis is tight, Guruswami-Kopparty~\cite{GuruswamiKopparty13} extended it in two ways.  The first is to study the \emph{multiplicities} and the second is to study the case when $t>r$.  They observed that the multiplicity of vanishing of $\det \Wr_r(\alpha)M$ upper bounds the rank deficiency $(\rank M-\rank \Wr_r(\alpha)M)$ so that instead of a weak condenser we can obtain a \emph{strong} condenser with the same list bound.  For $t>r$, they observed that taking the values of $\alpha$ as powers of $\omega$ introduces a certain \emph{redundancy} between different $\alpha$-values which is akin to the folding of folded Reed-Solomon codes of Guruswami-Rudra~\cite{GuruswamiRudra08} and this allows one to obtain a smaller list bound by pruning this redundancy. Updating their argument with the above degree bound (as powers of $\omega$ are non-zero, so we can safely ignore powers of $x$ in the determinant with respect to the determinant being non-zero) we get the following analysis of the folded Wronskian as a lossless rank condenser.

\foldedwronskian*

We now conclude with an explicit instantiation of the above when taking all points in a finite field.

\begin{corollary}\label{res:lossless-seeded_large-field}
	Let $\F_q$ be a finite field and let $n,t\ge r\ge 1$ such that $q>n$. Given a generator $\omega$ of $\F_q$, there is an explicit $\cE\subseteq\F_q^{t\times n}$ with $|\cE|=\floor{\frac{q-1}{t}}$, such that $\cE$ is a strong $\left(r,\frac{r(n-r)}{t-r+1}\right)$-lossless rank condenser. 
\end{corollary}
\begin{proof}
	As $\omega$ generates $\F_q$, we can apply \autoref{res:folded-wronskian-GK13} as $\omega$ has the desired multiplicative order of $q-1\ge n$.  Now take $S=\left\{1,\omega^t,(\omega^t)^2,\ldots,(\omega^t)^{\floor{\frac{q-1}{t}}}\right\}$ and $\cE\eqdef\{\Wr_t(\alpha)\st \alpha\in S\}$.  For any $1\le r\le t$, we have that $t\ge t-r+1$ so that $\cE$ has the desired condensing properties by \autoref{res:folded-wronskian-GK13}.  That $|\cE|=\floor{\frac{q-1}{t}}$ follows from construction by the multiplicative order of $\omega$. That $\cE$ is explicit is also clear as we can index $\cE$ by $\left[\floor{\frac{q-1}{t}}\right]$ and then given an $i\in\left[\floor{\frac{q-1}{t}}\right]$ produce $\Wr_t(\omega^i)$ in $\poly(n,t,\log q)$ operations in $\F_q$ via repeated squaring.
\end{proof}

\subsection{Lossy Rank Condensers}

We now turn from lossless condensers to \emph{lossy} condensers, as defined in \autoref{sec:lossless-seeded}. The motivation for studying objects that can lose a small amount of rank is to obtain a comparatively smaller seed length. Before turning to constructions, we study the issue of condensing ``rank $r$'' as compared to the stronger notion of condensing ``rank $\le r$'', where the latter notion is provably stronger.  We then relate lossless condensers to lossy condensers, observing that an averaging argument converts strong lossless condensers to lossy condensers with a \emph{smaller} list bound.  This leads us to constructing lossy condensers from our lossless constructions, and this will achieve condensing of ``rank $\le r$'' as needed for our application to dimension expanders (\autoref{sec:dim-exp_constr}) as they must expand \emph{all} small subspaces.

We begin by recalling that a $(r,\eps)$-lossy condenser should condenser rank $r$ to rank $(1-\eps)r$, and that a $(\le r,\eps)$-condenser should condense rank $s$ to rank $(1-\eps)s$ for \emph{all} $s\le r$.  Insisting on the latter notion is somewhat out of line with the usual definition of min-entropy condensers, that ask for min-entropy $\ge k$ being condensed to some $\ge k'$ min-entropy, without any (stated) guarantee on sources with input min-entropy $\ll k$. As such, $(\le r,\eps)$-condenser notion is more akin to the \emph{conductors} of Capalbo, Reingold, Vadhan and Wigderson~\cite{CapalboRVW02} (which are maps that have min-entropy output guarantees for any input min-entropy) than condensers. However, when $\eps=0$, we see that condensing for rank $r$ implies condensing $\le r$ with the same list bound, similar to \autoref{res:subspace design_les-vs-eqs}.

\begin{lemmawp}\label{res:lossless_le-v-eq}
	Let $\F$ be a field and $n\ge r\ge 1$. Let $\cE\subseteq\F^{t\times n}$ be a weak/strong $(r,L)$-lossless rank condenser.  Then $\cE$ is a $(\le r,L)$-lossless rank condenser.
\end{lemmawp}

However, we now note that this connection breaks for $\eps>0$.

\begin{lemma}\label{res:dim-le-r_dim-eq-r}
	Let $\F$ be a field, let $n\ge 1$. Let $\cE\subseteq\F^{t\times 3n}$ be a $(2n,\nicefrac{1}{2})$-lossy rank condenser.  Let $\pi:\F^{4n}\to\F^{3n}$ be the projection map onto the first $3n$ coordinates and let $P\in\F^{3n\times 4n}$ be the associated projection matrix. Then $\cE'\eqdef\{E P\st  E\in\cE\}\subseteq\F^{t\times 4n}$ is a $(3n,\nicefrac{2}{3})$-lossy condenser but not a $(s,1-\delta)$-condenser for any $1\le s\le n$ and $\delta<1$.
\end{lemma}
\begin{proof}
	For a vector space $V\subseteq\F^{4n}$, the dimension of the projection $\pi(V)$ has $\dim \pi(V)=\dim V- \dim (V\cap\ker\pi)\ge \dim V-n$.  Thus, if $\dim V\ge 3n$ then $\dim \pi(V)\ge 2n$, so that there is some $E\in\cE$ so that $E(\pi(V))$ has dimension $\ge n$.  Thus $\cE'$, the composition of $\cE$ and $\pi$, will map spaces of dimension $3n$ to spaces of dimension $\ge n$ so that $\cE'$ is a $(3n,\nicefrac{2}{3})$-condenser.

	Now consider $V\subseteq\F^{4n}$ of dimension $s\le n$ which is a subspace of the kernel of the projection map $\pi$ (which has dimension $n$).  Then clearly $\pi(V)=0$, so that the rank of the condenser $\cE'$ on $\pi(V)$ is always zero, so $\cE'$ preserves none of the rank of $V$, so that $\cE'$ is not a $(s,1-\delta)$-condenser for any $1\le s\le n$ and $\delta<1$.
\end{proof}

This example embeds itself into many examples of ``manipulating rank $r$'' versus ``manipulating rank $\le r$'' by pseudorandom objects.  Thus, it shows that to obtain the latter guarantee one needs to explicitly consider dimension $\le r$, and indeed our techniques will work in this regime.  

We now give constructions of good lossy rank condensers, using \emph{strong} lossless rank condensers.  As strong lossless condensers bound the sum of \emph{all} rank losses $(\rank M-\rank EM)$ it follows from an averaging argument that an at most $\nicefrac{1}{k}$ fraction of the maps can have a rank deficiency $\ge k$. Thus we can take a seed length of the lossy condenser that is $\nicefrac{1}{k}$ of the list bound of the lossless condenser. For this reduction, a weak lossless condensers would not suffice as the resulting seed length would not be smaller than the original list bound.

\begin{proposition}\label{res:strong-lossless-to-lossy}
	Let $\F$ be a field and let $n\ge r\ge 1$ and $\eps\ge 0$. Let $\cE\subseteq\F^{t\times n}$ be a strong $(r,L)$-lossless rank condenser.  Then for any $\cE'\subseteq \cE$ with $|\cE'|>\frac{L}{\floor{\eps r}+1}$, $\cE'$ is a $(r,\eps)$-lossy rank condenser.
\end{proposition}
\begin{proof}
	We argue the contrapositive. As $\cE'$ is not a lossy condenser there is a matrix $M\in\F^{n\times r}$ with rank $r$ so that for all $E\in\cE'$,
	\[
		\rank EM< (1-\eps)\rank M
		\;,
	\]
	and thus
	\[
		\rank M -\rank EM>\eps\rank M=\eps r
		\;,
	\]
	so that
	\[
		\rank M -\rank EM\ge\floor{\eps r}+1
		\;.
	\]
	Thus, using that $\cE'$ is a strong $(r,L)$-design (as it is a subset of $\cE$ is such a design, and this is preserved under taking subsets),
	\begin{align*}
		L
		&\ge \sum_{E\in \cE'}(\rank M -\rank EM)\\
		&\ge \sum_{E\in \cE'}(\floor{\eps r}+1)\\
		&=|\cE'|\cdot(\floor{\eps r}+1)
		\;,
	\end{align*}
	and thus $|\cE'|\le \frac{L}{\floor{\eps r}+1}$ as desired.
\end{proof}

We now use this lemma, along with the results on lossless condensers, to obtain our desired lossy condenser. 

\begin{proposition}\label{res:gk13-lossy-condenser}
	Assume the setup of \autoref{constr:folded-wronskian} where we take $n,t\ge r\ge 1$ and $\eps>0$. Let $S\subseteq \{(\omega^\ell)^j\st j\ge 0\}$ where $\ell\ge t$ and $|S|\ge \min\left\{\frac{n}{\eps(t-r+1)},n^2 \right\}$.  Then $\cE\eqdef\{\Wr_t(\alpha)\st \alpha\in S\}\subseteq\F^{t\times n}$ is a $(\le r,\eps)$-lossy rank condenser.
\end{proposition}
\begin{proof}
	If $|S|\ge n^2\ge r(n-r)+1$ then we have a \emph{lossless} $(\le r,0)$-condenser by \autoref{res:lossless-extractor} (and \autoref{res:lossless_le-v-eq}).

	Thus consider $|S|\ge \frac{n}{\eps(t-r+1)}$. By \autoref{res:folded-wronskian-GK13} and \autoref{res:strong-lossless-to-lossy} it follows that $\cE$ is a $(s,\eps)$-lossy condenser as long as $|S|>\frac{s(n-s)}{(\floor{\eps s}+1)(t-s+1)}$ (as $\ell\ge t\ge t-s+1$ as $s \ge 1$).  In particular, it suffices if 
	\begin{align*}
		|S|&\ge \frac{sn}{(\floor{\eps s}+1)(t-s+1)}
		\intertext{and thus as $\floor{\eps s}+1\ge \eps s$ and $s\le r$, it suffices if}
		|S|&\ge \frac{sn}{\eps s\cdot (t-r+1)}=\frac{n}{\eps \cdot (t-r+1)}
		\;,
	\end{align*}
	where this last bound is independent of $s$, so that we get the desired result.
\end{proof}

Note that in the above it is crucial that not only does the condenser of \autoref{res:folded-wronskian-GK13} condense all small ranks $s\le r$, but also that the list bound is smaller as $s$ decreases. Now we take this construction with explicit values for the $\alpha$'s.

\begin{corollary}\label{res:gk13-lossy-condenser_explicit}
	Let $\F$ be a field. Let $n,t\ge r\ge 1$ and $\eps>0$. Define $N\eqdef \ceilp{\frac{n}{\eps(t-r+1)}}$ and let $M\ge N$. Suppose $|\F|$ is of size $>tn^2$. Then there is an explicit $(\le r,\eps)$-lossy rank condenser $\cE\subseteq\F^{t\times n}$ of size $|\cE|=M$.
\end{corollary}
\begin{proof}
	As $|\F|>tn^2$ it follows that we can find an element $\omega\in\F$ with multiplicative order $\ge tn^2$ in $\poly(n,t)$ steps (see for example Forbes~\cite[Lemma A.0.5]{Forbes14}).  Thus, the set $S\eqdef \{(\omega^t)^j|0\le j<\min\{N,n^2\}\}$ has size $\min\{N,n^2\}$ and these explicit elements are all distinct as $\omega$ has order $\ge t\min\{N,n^2\}$.

	We now appeal to \autoref{res:gk13-lossy-condenser}, seeing that our set $S$ is sufficiently large and our setup matches that of \autoref{constr:folded-wronskian} as $\omega$ has order $\ge n$.  Further, we see that the resulting matrices $\cE$ are explicit as \autoref{constr:folded-wronskian} is explicit. Padding the result with zero-matrices yields an explicit set $\cE$ with size $M$.
\end{proof}

Thus, we see that the above essentially almost matches the list bound of $\frac{n}{\eps(t-(1-\eps)r)}$ of the probabilistic method for lossy condensers condensing dimension $\le r$, as given by \autoref{res:lossy-seeded_prob-method}.  However, there are two gaps.  The first is that this construction requires the use of polynomially large fields, while the existential bound also holds for constant-sized fields.  Second is that this construction requires that the output $t$ have ``$t\ge r$'' even though we only require the rank to be $\ge(1-\eps)r$.  The probabilistic method of \autoref{res:lossy-seeded_prob-method} allows us to take $t\approx (1-\eps)r$.  While this does not seem dramatic, it will cause a slight complication in our construction of dimension expanders (as discussed after \autoref{res:tensor-then-condense_instantiate_gamma0}).

\section{Constructions of Dimension Expanders}\label{sec:dim-exp_constr}

In this section we construct constant-degree dimension expanders by composing a tensoring operation with our construction of a lossy rank condenser, as constructed in in \autoref{sec:lossless-seeded_constr}.  We first discuss previous constructions of dimension expanders, then turn to our own construction.

\subsection{Previous Constructions} 

There have been two main types of constructions.  The first is to use Cayley graphs from groups with Kazhdan's \emph{property T}, and the second is to use \emph{monotone expanders}.

\smallskip \noindent {\bf Property $T$:} This approach to constructing dimension expanders is rooted in their similarities to expanding Cayley graphs, which we now discuss.  Specifically, dimension expanders can be seen as a certain type of (vertex) expander in the usual sense, as observed by Dvir and Shpilka~\cite{DvirShpilka11}.  That is, given a $(\Omega(1),1+\Omega(1))$-dimension expander $\cA=\{A_1,\ldots,A_d\}\subseteq\F_q^n$ (where we assume without loss of generality that all $A_i$ are invertible, as discussed after \autoref{defn:dim-exp}), consider the graph $G$ with vertex set $\F_q^n$ such that $\vv\in\F_q^n$ is connected to $\{A_i \vv\}_i$.  Note that this is a Schreier graph, as we have a subgroup of $\GL_n(\F_q)$ (the subgroup generated by the $A_i$) acting on the set $\F_q^n$. While in general this graph $G$ is directed, one could assume (as is common in Cayley and Schreier graphs) that $\cA$ is symmetric so that $A\in \cA$ iff $A^{-1}\in \cA$, in which case $G$ is undirected.  That the graph $G$ expands (as a vertex expander) would mean that whenever $S\subseteq\F_q^n$ has $|S|\le (1-\Omega(1))|\F_q^n|$, that the neighborhood of $S$ has size at least $1+\Omega(1)$ times the size of $S$.  That is, that $|\cup_i A_i(S)|\ge (1+\Omega(1))|S|$.  That $\cA$ is a \emph{dimension} expander asks for an expansion property of $G$ that is both weaker and stronger than that of vertex expansion in some respects.  It is weaker in that we only care about when the set $S$ is a \emph{subspace} of $\F_q^n$.  However, it is also stronger as vertex expansion only yields that $|\cup_i A_i(S)|\ge (1+\Omega(1))|S|$ which only implies $\dim \spn \cup_i A_i(S)\ge \dim S+\Omega(1)$, while dimension expansion yields that $\dim \spn \cup_i A_i(S)\ge (1+\Omega(1))\dim S$.

As seen by the above connection, dimension expanders can be seen as a type of Schreier graph.  Thus, to understand the construction of dimension expanders it is first helpful to recall the construction of expanding Cayley graphs in particular (as these are Schreier graphs). In particular, if you have a group $G$ with generators $\Gamma$, then the corresponding Cayley graph is an expander if the \emph{Kazhdan constant} is strictly bounded away from zero (in which case it is said that $G$ has \emph{property $T$} with respect to $\Gamma$).  The Kazhdan constant being bounded away from zero roughly means that each irreducible unitary representation of $G$ must ``move'' each non-zero vector a non-trivial amount via some generator in $\Gamma$.

Given that the above notions are inherently linear algebraic, Wigderson~\cite{Wigderson04} made a conjecture (see Dvir and Wigderson~\cite[Conjecture 7.1]{DvirWigderson11}) that any expanding Cayley graph would yield a dimension expander.  Specifically, that any irreducible representation $\rho:G\to\F^{n\times n}$ of that group $G$ with generators $\Gamma\subseteq G$ would have that $\rho(\Gamma)$ is a dimension-expander.  This collection will be of constant-size if the original expander was constant degree, and would intuitively expand dimension by analogy to how a positive Kazhdan constant ``moves'' any vector so these matrices must ``move'' a subspace to obtain a non-trivial amount of additional dimension.

Lubotzky and Zelmanov~\cite{LubotzkyZelmanov08} proved Wigderson's conjecture in characteristic zero by exploiting the connection of expansion in Cayley graphs to the underlying group having property $T$. Thus, they established explicit constant-degree $(\Omega(1),1+\Omega(1))$-dimension expanders in any field of characteristic zero. Unfortunately, as property $T$ relies on the representations being over the characteristic zero (so that distance is a well-defined notion), the conjecture remains open in finite characteristic. Harrow~\cite{Harrow08} independently obtained this result in the context of \emph{quantum expanders}, which imply dimension expanders in characteristic zero.  We summarize this in the following theorem.

\begin{theoremwp*}[Lubotzky and Zelmanov~\cite{LubotzkyZelmanov08} and Harrow~\cite{Harrow08}]
	Let $\F$ be a field of characteristic zero and $n\ge 1$. There exists an explicit $\O(1)$-sized collection $\cA\subseteq\F^{n\times n}$ such that $\cA$ is a $(\nicefrac{1}{2},1+\Omega(1))$-dimension expander over $\F^n$.
\end{theoremwp*}

\smallskip \noindent {\bf Monotone Expanders:} This second approach to constructing dimension expanders exploits another similarity to expander graphs, but now to bipartite (vertex) expanders.  Specifically, as observed in Bourgain and Yehudayoff~\cite{BourgainYehudayoff13}, bipartite vertex expanders can be seen as special cases of dimension expanders, where bipartite vertex expanders only expand subspaces spanned by basis vectors and dimension expanders expand all subspaces.  Indeed, suppose the graph $G$ is on the vertex set with bipartition $[n]\sqcup[n]$ and we partition the edges $E$ of $G$ into partial matchings $E=E_1\sqcup\cdots\sqcup E_d$ so that $d$ is an upper bound on the degree of $G$.  Then we can view the sets $E_i$ as defining (partial) maps $E_i:[n]\to[n]$ which we can then view as matrices $A_i\in\bits^{n\times n}$ by using these maps to act on the standard basis vectors and then extending linearly.  That $G$ is a good vertex expander means that for any $S\subseteq[n]$ with $|S|\le (1-\Omega(1))n$ that the neighborhood of $S$ in $G$ is slightly larger, that is, $|\cup_i E_i(S)|\ge (1+\Omega(1))|S|$.  However, from this we see that the vector space $V\eqdef\spn\{\ve_i\}_{i\in S}$ (where $\ve_i$ is the $i$-th standard basis vector) thus expands under the collection $\{A_i\}$, as $\dim \sum_i A_i(V)=|\cup_i E_i(S)|\ge (1+\Omega(1))|S|=(1+\Omega(1))\dim V$.

Given the above connection, one can then ask in which contexts do the above matrices $A_i$ also expand \emph{any} subspace, as opposed to just those spanned by basis vectors.  Dvir and Shpilka~\cite{DvirShpilka11} implicitly observed that dimension expansion occurs when these (partial) maps $E_i:[n]\to[n]$ are \emph{monotone}. That is, if each edge set $E_i$ defines a partial map $E_i:[n]\to[n]$ so that if $E_i(j)$ and $E_i(k)$ are defined for $j,k\in[n]$, then $j<k\implies E_i(j)<E_i(k)$.  When this monotonicity occurs, the resulting matrices $\{A_i\}_i$ are a dimension expander, and this statement was made explicit by Dvir and Wigderson~\cite{DvirWigderson11}.

Thus, to construct explicit constant-degree dimension expanders it then suffices to construct explicit constant-degree monotone expanders (where the partition of the edges into monotone maps must also be explicit).  Unfortunately, monotone expanders seem a more delicate object than unrestricted expanders.  Indeed, the standard probabilistic method arguments cannot demonstrate even the \emph{existence} of constant-degree monotone expanders (see \cite{DvirWigderson11,BourgainYehudayoff13}). However, using this connection along with Cayley expanders over $\Z_n$, Dvir and Shpilka~\cite{DvirShpilka11} were able to construct monotone expanders (and thus dimension expanders) with logarithmic degree, as well as constant-degree expanders with inverse-logarithmic expansion.  We formally state their result here to contrast with our results over $\F_2$ (\autoref{res:tensor-then-condense_instantiate_small-field}), as our results are much weaker (only achieving $(\Omega(\nicefrac{1}{\lg n}),1+\Omega(1))$-expanders of logarithmic degree).

\begin{theoremwp*}[Dvir and Shpilka~\cite{DvirShpilka11}]
	Let $n\ge 1$. There exists an explicit $\O(\lg n)$-sized collection $\cA\subseteq\bits^{n\times n}$ such that $\cA$ is a $(\Omega(1),1+\Omega(1))$-dimension expander over $\F^n$, for every field $\F$.

	Let $n\ge 1$. There exists an explicit $\O(1)$-sized collection $\cA\subseteq\bits^{n\times n}$ such that $\cA$ is a $(\Omega(1),1+\Omega(\nicefrac{1}{\lg n}))$-dimension expander over $\F^n$, for every field $\F$.
\end{theoremwp*}

Dvir and Wigderson~\cite{DvirWigderson11} gave an iterative construction of monotone expanders in spirit of the zig-zag product of Reingold, Vadhan and Wigderson~\cite{ReingoldVadhanWigderson02}. Using this approach they were able to give monotone expanders (and thus dimension expanders) of degree $\lg^{(c)}(n)$ (the $c$-th iterated logarithm) for any constant $c$.  Assuming a base construction of a constant-degree monotone expander (which is \emph{not} known to exist via the probabilistic method), they were able to produce constant-degree dimension expanders.

In a more sophisticated work, Bourgain and Yehudayoff~\cite{BourgainYehudayoff13} used expansion in the group $\SL_2(\R)$ to obtain explicit constant degree monotone expanders, and thus dimension expanders.

\begin{theoremwp*}[Bourgain and Yehudayoff~\cite{BourgainYehudayoff13}]
	Let $n\ge 1$. There exists an explicit $\O(1)$-sized collection $\cA\subseteq\bits^{n\times n}$ such that $\cA$ is a $(\nicefrac{1}{2},1+\Omega(1))$-dimension expander over $\F^n$, for every field $\F$.
\end{theoremwp*}

A remarkable feature of the above works is that they achieve (for each $n$) a \emph{single} collection of matrices that is a dimension expander for \emph{every} field, in particular by only using $\bits$-values. In contrast, our constructions will very much depend on the underlying field by using elements of large multiplicative order.

\subsection{Our Construction}

We now proceed to give the details of our construction, following the outline given in \autoref{sec:dim-exp}.  That is, we apply a tensoring operation to yield expansion but increasing the ambient dimension, and then use a lossy rank condenser to preserve this expansion while reducing the ambient dimension to its original size.  As mentioned in \autoref{sec:dim-exp}, this approach is somewhat limited to expanding rank $\le\nicefrac{n}{2}$ subspaces as we can only tensor with integral-dimensional spaces.  To circumvent this, we observe that we can simply ``forget'' some of the rank we obtained by tensoring and this allows the construction to expand any rank.  We start with a generic analysis, where we parameterize the ``forgetfullness'' by $\gamma$.

\begin{proposition}\label{res:tensor-then-condense}
	Let $\F$ be a field. Let $n\ge r\ge 1$, $d\ge 1$, $\delta\in\ocint{0}{1}$, and $\gamma\in[0,1]$. Let $\cE\subseteq\F^{n\times nd}$ be a $(\le \ceil{(1-\gamma)rd},\delta)$-lossy rank condenser.  For $i\in[d]$, define $T_i\in \F^{nd\times n}$ to be the matrix of the map $\vv\mapsto \vv\otimes \ve_i$, where $\ve_i\in\F^{d}$ is the $i$-th standard basis vector.  Define $\cA\eqdef\{ E T_i\st E\in\cE, i\in[d]\}$.  Then $\cA\subseteq\F^{n\times n}$ is a $(\nicefrac{r}{n},(1-\gamma)(1-\delta)d)$-dimension expander of degree $d\cdot |\cE|$.
\end{proposition}
\begin{proof}
	That $\cA\subseteq\F^{n\times n}$ is clear from construction, as $E\in\F^{n\times nd}$ and $T_i\in\F^{nd\times n}$.  That $|\cA|=d\cdot |\cE|$ is also clear.  We now argue the expansion property.

	Consider some $V\subseteq\F^n$ with $\dim V=s\le r$. Then we have that $V\otimes \F^d\subseteq\F^{nd}$ has rank $sd$ by the properties of the tensor product, and using that $V\otimes \F^d=\sum_i T_i(V)$, it follows that $\dim\sum_i T_i(V)=sd$.  In particular, we have that there is some subspace $W\subseteq \sum_i T_i(V)$ with $\dim W=\ceil{(1-\gamma)sd}\le \ceil{(1-\gamma)rd}$.  By the hypothesis on $\cE$, it follows that for some $E\in\cE$ that $\dim E(W)\ge (1-\delta)\ceil{(1-\gamma)sd}\ge (1-\delta)(1-\gamma)sd=(1-\delta)(1-\gamma)d\cdot \dim V$.  Thus, as $\dim \sum_{A\in\cA} A(V)\ge \dim E(W)\ge (1-\delta)(1-\gamma)d\cdot \dim V$ we see that there is the desired expansion.
\end{proof}

We now instantiate this recipe with our explicit construction of a lossy rank condenser (\autoref{res:gk13-lossy-condenser_explicit}) to deduce the following.

\begin{theorem}\label{res:tensor-then-condense_instantiate}
	Let $\F$ be a field. Let $n,d\ge 1$. Let $\eps\in(0,1)$, $\delta\in\ocint{0}{1}$ and $\gamma\in[0,1]$, subject to
	\[
		(1-\gamma)\eps d < 1
		\;.
	\]
	Then there is an explicit $(\eps,(1-\gamma)(1-\delta)d)$-dimension expander of degree 
	\[
		d\cdot \ceilp{\frac{d}{\delta (1-(1-\gamma)\eps d)}}
		\;,
	\]
	whenever $|\F|>d^2n^3$.
\end{theorem}
\begin{proof}
	For $r=\floor{\eps n}$, \autoref{res:gk13-lossy-condenser_explicit} yields an explicit $(\le \ceil{(1-\gamma)rd},\delta)$-lossy rank condenser of size $M\eqdef \ceilp{\frac{d}{\delta (1-(1-\gamma)\eps d)}}$ as
	\[
		\ceilp{\frac{nd}{\delta (n-\ceil{(1-\gamma)rd}+1)}}
		\le \ceilp{\frac{nd}{\delta (n-(1-\gamma)rd)}}
		\le \ceilp{\frac{d}{\delta (1-(1-\gamma)\eps d)}}
		\;.
	\]
	Note that \autoref{res:gk13-lossy-condenser_explicit} requires that $n\ge \ceil{(1-\gamma)rd}$, which is true iff $(1-\gamma)rd\le n$, which follows from $(1-\gamma) \eps d\le 1$, and we have by hypothesis that $(1-\gamma)\eps d < 1$.  We now use this condenser in \autoref{res:tensor-then-condense}, noting that this construction is also explicit and multiplies the degree by $d$. 
\end{proof}

We now note one particularly natural set of parameters for the above construction, namely when $\gamma=0$.

\begin{corollarywp}\label{res:tensor-then-condense_instantiate_gamma0}
	Let $\F$ be a field. Let $n,d\ge 1$. Let $\eps\in(0,1)$, $\delta\in\ocint{0}{1}$, subject to $\eps <\nicefrac{1}{d}$. Then there is an explicit $(\eps,(1-\delta)d)$-dimension expander of degree $d\cdot \ceilp{\frac{d}{\delta (1-\eps d)}}$ whenever $|\F|>d^2n^3$.
\end{corollarywp}

Now consider this construction when $\eps\le\nicefrac{1}{2d}$ and $\delta=\nicefrac{1}{2}$, so that we have that the degree of the expander is $\approx 4d^2$, when the expansion $\alpha$ has $\alpha=\nicefrac{d}{2}$. However, the existential construction of \autoref{res:dim-exp_prob-method} yields that the degree could be $\approx \alpha+2$ here.  Thus, our construction suffers from a quadratic loss as compared to the existential bounds.  In particular, we do not achieve lossless dimension expanders. Intuitively, this quadratic loss is because we compose a tensor step with a condensing step, so that the degree of each step multiplies. 

While the above corollary with $\gamma=0$ is sufficient for expanding small rank subspaces, it intrinsically cannot yield, for example, $(\nicefrac{2}{3},1+\Omega(1))$-expanders, because the parameter $d$ must be an \emph{integer}. Thus, after our tensoring step, from rank $r$ we get rank $rd$ and we must take $d>1$ to achieve \emph{any} expansion.  But as there is no way to tensor with $\F^d$ for $d$ fractional, we must have that $d\ge 2$. However, our construction of a lossy rank condenser (\autoref{res:gk13-lossy-condenser}) \emph{requires} that the output size be at least the rank bound of $rd\ge 2r$.  Thus, since the output size must be $n$ so that we achieve $n\times n$ matrices, it follows that this method does not work for $2r>n$.  Note that if we had lossy condensers meeting the existential bound (\autoref{res:lossy-seeded_prob-method}) then we would only need that $n\ge (1-\delta)rd$ and thus taking $\delta$ sufficiently close to 1 would remedy the $r\le \nicefrac{n}{2}$ limitation seen here.

However, as we currently cannot match the above existential bound, the above construction introduces the $\gamma$ parameter.  When $\gamma>0$, we simply ``forget'' that the tensoring yields dimension $rd$ and simply work with the smaller rank bound of $(1-\gamma)rd$, thus allowing us to take any $r$ where $(1-\gamma)rd\le n$.  In particular, by letting $\gamma$ close to $1$ we can take \emph{any} rank $r$.  We now implement the above by choosing $\gamma$ carefully so that we can obtain constant degree expanders in $\F^n$ that expand rank $\eps n$ to rank $\eta n$ for any constants $\eps\le \eta<1$. 

\begin{corollary}\label{res:tensor-then-condense_instantiate_gamma-pos}
	Let $\F$ be a field. Let $n\ge 1$. Let $0<\eps\le \eta<1$. Then there is an explicit $(\eps,\nicefrac{\eta}{\eps})$-dimension expander of degree 
	\[
		\ceilp{\frac{1+\eta}{2\eps}}\cdot \ceilp{\frac{2(1+\eta)\ceilp{\frac{1+\eta}{2\eps}}}{(1-\eta)^2}}
		\le
		\ceilp{\nicefrac{1}{\eps}}\cdot \ceilp{\frac{4\ceilp{\nicefrac{1}{\eps}}}{(1-\eta)^2}}
		\;,
	\]
	whenever $|\F|>d^2n^3$.
\end{corollary}
\begin{proof}
	This follows from \autoref{res:tensor-then-condense_instantiate} by choosing parameters carefully.  In particular, we will to choose $d,\delta,\gamma$ so that $(1-\delta)(1-\gamma)\eps d=\eta$ but that $(1-\gamma)\eps d$ is bounded away from one.

	In particular, choose $d=\ceilp{\frac{1+\eta}{2\eps}}$.  Note that as $\eps<1$ we have that $d\ge 2$.  Now choose $\gamma$ so that $(1-\gamma)\eps d = \frac{1+\eta}{2}$ so that $(1-\gamma)\ceilp{\frac{1+\eta}{2\eps}}= \frac{1+\eta}{2\eps}$ from which it follows that $\gamma\in\coint{0}{1}$. Now choose $\delta$ so that $(1-\delta)(1-\gamma)\eps d=\eta$, in particular that $(1-\delta)=\frac{2\eta}{1+\eta}$, so that $\delta=\frac{1-\eta}{1+\eta}\in(0,1)$.  Thus, plugging these values into \autoref{res:tensor-then-condense_instantiate_gamma0} yields the desired parameters.
\end{proof}

The above constructions all require large fields.  In \autoref{sec:small-fields} we show how to simulate these results in small fields by paying certain logarithmic penalties.

\section{Constructions over Small Fields}\label{sec:small-fields}

The main constructions of this paper rely on the folded Wronskian (\autoref{constr:folded-wronskian}) which requires a polynomially large field.  In this section, we discuss to what extent we can extend our techniques to smaller fields.  Guruswami and Kopparty~\cite{GuruswamiKopparty13} gave a way to convert subspace designs over large fields to subspace designs over small fields with comparable parameters.  However, their method was not able to preserve the \emph{strong-ness} of the subspace design and as we saw this strong-ness is essential for our construction of constant degree dimension expanders (\autoref{sec:dim-exp_constr}).  We give in this section an alternate method for simulating large fields that preserves strong-ness but is slightly worse in other parameters than the method of Guruswami and Kopparty~\cite{GuruswamiKopparty13}.  Our method is based on the conversion of Reed-Solomon codes to BCH codes and can be seen as a basic form of ``code concatenation'' from coding theory (although more sophisticated versions of that idea do not seem to work in our setting, see \autoref{rmk:code-concate}).  As a consequence, we obtain over \emph{any} field constructions of strong lossless rank condensers with ``inverse logarithmic output rate'' and logarithmic-degree $(\Omega(\nicefrac{1}{\lg n}),1+\Theta(1))$-dimension expanders. 

We begin by showing how to transform a matrix over an extension field $\K$ of $\F$ to a matrix over $\F$, while preserving rank.  This operation will increase the rows of this matrix (which is undesirable, but tolerable) and is akin to converting Reed-Solomon codes to BCH codes.

\begin{construction}\label{const:BCH}
	Let $\F$ be a subfield of $\K$, where $\dim_\F\K=k$ so that $\varphi:\K\cong\F^k$ is an $\F$-vector space isomorphism.  Define $\varphi^n:\K^n\cong\F^{kn}$ by applying $\varphi$ coordinate-wise, so that $(\vaa_1,\ldots,\vaa_n)\mapsto (\varphi(\vaa_1),\ldots,\varphi(\vaa_n))$.  For a matrix $M\in\K^{n\times m}$, define $\varphi^n(M) \in \F^{kn\times m}$ as the result of applying $\varphi^n$ to each column of $M$.
\end{construction}

The key point about this map $\varphi^n$ is that it is $\F$-linear so that it composes nicely with matrix multiplication.

\begin{lemma}\label{res:BCH-compose}
	Assume the setup of \autoref{const:BCH}.  Let $E\in\K^{t\times n}$ and $M\in\F^{n\times r}$.  Then $\varphi^t(EM)=\varphi^t(E)M$.
\end{lemma}
\begin{proof}
	Let $E$ have columns $\vv_1,\ldots,\vv_n\in\K^t$.  Then the $i$-th column is $(EM)_{\bullet,i}=\sum_{j=1}^n M_{j,i}\vv_j$.  Thus by $\F$-linearity of $\varphi^t$, $\varphi^t((EM)_{\bullet,i})=\sum_{j=1}^n M_{j,i}\varphi^t(\vv_j)=(\varphi^t(E) M)_{\bullet,i}$.  Thus, each column of $\varphi^t(EM)$ and $\varphi^t(E)M$ agree, so they are equal.
\end{proof}

We now show that this map $\varphi^n$ preserves rank of matrices as we switch fields.

\begin{lemma}\label{res:BCH-simulation}
	Assume the setup of \autoref{const:BCH}. Let $M\in\K^{n\times m}$ be a matrix.  Then $\varphi^n(M)$ has $\rank_\F \varphi^n(M)\ge \rank_\K M$.
\end{lemma}
\begin{proof}
	Consider the kernel of the matrix $M$ over $\K$, $\ker_\K M=\{\vv\st \vv\in\K^m , M\vv=\vec{0}\}$, so that this represents the $\K$-dependencies among the columns of $M$.  Likewise, consider the  $\F$-dependencies of $\varphi^n(M)$, $\ker_\F \varphi^n(M)=\{\vv\st \vv\in\F^m , \varphi^n(M)\vv=\vec{0}\}$.  By \autoref{res:BCH-compose} we have  $\varphi^n(M\vv)=\varphi^n(M)\vv$ for $\vv\in\F^m$. As $\varphi^n$ is an isomorphism, we conclude that for $\vv\in\F^m$, $M\vv=\vec{0}$ iff $\varphi^n(M)\vv=\vec{0}$.  Thus, it follows that $\ker_\F \varphi^n(M)\subseteq \ker_\K M$.  In particular, we see that the $\F$-dependencies between the columns of $M$ are a subset of the $\K$-dependencies of the columns of $M$, from which the claim follows.
\end{proof}

Thus, we arrive at the following.

\begin{corollarywp}\label{res:BCH-compose-rank}
	Assume the setup of \autoref{const:BCH}.  Let $E\in\K^{t\times n}$ and $M\in\F^{n\times r}$.  Then $\rank_\F \varphi^t(E)M\ge\rank_\K EM$.
\end{corollarywp}

Note that in the above the results are in general neither tight nor improvable.  To see that the rank can increase, consider $\alpha_1,\ldots,\alpha_k\in\K$ which are an $\F$-basis for $\K$. Then the matrix $M=[\alpha_1,\ldots,\alpha_k]\in\K^{1\times k}$ has $\K$-rank equal to 1, while $\varphi^1(M)$ has $\F$-rank equal to $k$.  To see that the rank may not increase, consider the $n\times n$ identity matrix $\Id_n\in\K^n$. The resulting matrix $\varphi^n(\Id_n)\in\F^{kn\times n}$ still has rank $n$ as its rank cannot exceed the number of columns.

In general, one can ``get more'' $\F$-rank from a matrix over $\K$ by expanding the number of rows \emph{and columns} of this matrix.  That is, for a $\K$-vector space $V$ we use that $\dim_\F V=k\dim_\K V$ instead of just that $\dim_\F V\ge\dim_\K V$. Indeed, Guruswami-Kopparty~\cite{GuruswamiKopparty13} used this to simulate large fields in their construction of subspace designs.  The downside of this approach is that it does not act as the composition of linear maps (so we do not get \autoref{res:BCH-compose}), and thus does not preserve strong-ness of the design.  

We now apply the above results to convert strong lossless rank condensers over large fields to condensers over small fields.

\begin{proposition}\label{res:BCH-simulate-cond}
	Assume the setup of \autoref{const:BCH} with $n\ge r\ge 1$. Let $\cE\subseteq\K^{t\times n}$ be a weak/strong $(r,L)$-lossless rank condenser.  Then $\varphi^n(\cE)\eqdef\{\varphi^n(E)\st E\in\cE\}\subseteq\F^{kt\times n}$ is a weak/strong $(r,L)$-lossless rank condenser.
\end{proposition}
\begin{proof}
	Consider a matrix $M\in\F^{n\times r}$ with $\rank_\F M=r$.  Then by \autoref{res:BCH-compose-rank}, we have that for any $E\in\cE$ that $\rank_\F\varphi^t(E)M\ge \rank_\K EM$.  Noting that $\rank_\F M=\rank_\K M$ as $M$ is an $\F$-matrix, we have the inclusion of sets $\{E\st \rank_\K EM<\rank_\K M\}\supseteq\{E\st \rank_\F \varphi^t(E)M<\rank_\F M\}$, so since the former set has size $\le L$ then so does the latter. As this holds for any $M$, it follows that $\varphi^t(\cE)$ has the desired weak condensing.  
	
	Similarly, we see that $\rank_\F M-\rank_\F\varphi^t(E) M \le \rank_\K M-\rank_\K EM$ and since summing over $E\in\cE$ of the right-hand side yields $\le L$ (assuming now $\cE$ is strong), then so does summing over the left-hand side.  Thus, $\varphi^t(\cE)$ has the desired strong condensing as well.
\end{proof}

We now apply the above simulation of a large field to translate \autoref{res:lossless-seeded_large-field} to small fields.

\begin{corollary}\label{res:lossless-seeded_small-field}
	Let $\F_q$ be a finite field and let $n,t\ge r\ge 1$ such that $q^k>n$. Given an explicit presentation of $\F_{q^k}\cong\F_q^k$ and a generator $\omega$ of $\F_{q^k}$, there is an explicit $\cE\subseteq\F_q^{kt\times n}$ with $|\cE|=\floor{\frac{q^k-1}{t}}$, such that for all $1\le r\le t$, $\cE$ is a strong $(r,\frac{r(n-r)}{t-r+1})$-lossless rank condenser (and thus strong $(r,\frac{r(n-r)}{t-r+1})$-subspace design).
\end{corollary}
\begin{proof}
	\autoref{res:lossless-seeded_large-field} implies that we have such a $\cE'\in\F_{q^k}^{t\times n}$ with $|\cE'|=\floor{\frac{q^k-1}{t}}$ that has the desired condensing properties (and thus design properties by our equivalence (\autoref{res:subspace-design-equal-lossless})).  We now apply our large field simulation (\autoref{res:BCH-simulate-cond}) to obtain $\cE$.  Clearly $\cE$ has the desired size and condensing properties.  That $\cE$ can be indexed in $\poly(n,t,k\log q)$ is also clear given the explicit presentation $\F_{q^k}\cong\F_q^k$.
\end{proof}

In comparison, Guruswami and Kopparty~\cite{GuruswamiKopparty13} obtain a similar construction of weak designs where the list bound has \emph{decreased} by a factor of $k$. However, for our applications the strong-ness of the above is more important.

Similarly, we can obtain explicit lossy rank condensers where the output size is logarithmically larger than it was previously.

\begin{corollary}\label{res:lossy-cond_small-field}
	Let $\F_q$ be a finite field.  Let $n,t\ge r\ge 1$ and $\eps>0$. Define $N\eqdef \ceilp{\frac{n}{\eps(t-r+1)}}$ and let $M\ge N$. Then there is an explicit $(\le r,\eps)$-lossy rank condenser $\cE'\subseteq\F^{kt\times n}$ of size $|\cE'|=M$, where $k\eqdef \ceil{\log_q(tn^2+1)}=\Theta(\log_q tn)$.
\end{corollary}
\begin{proof}
	First, observe that $tn^2+1\le q^k\le q\cdot (tn^2+1)\le (tn^2+1)^2$.  Now, recall that we can construct an explicit presentation of $\F_{q^k}$ extending $\F_q$ in time $\poly(q^k)\le\poly(n,t)$ (see for example Forbes~\cite[Lemma 3.2.5]{Forbes14}).  Now by \autoref{res:gk13-lossy-condenser_explicit} we have a $(\le r,\eps)$-lossy rank condenser $\cE\subseteq\F_{q^k}^{t\times n}$ of size $|\cE|=M$. We then convert this to be over $\F_q$ following the logic of \autoref{res:BCH-simulate-cond}.
\end{proof}

Likewise we obtain logarithmic-degree dimension expanders, but we can only expand rank of inverse logarithmic rate.

\begin{corollary}\label{res:tensor-then-condense_instantiate_small-field}
	Let $\F_q$ be a finite field. Let $n,d\ge 1$ and define $k\eqdef \ceil{\log_q(d^2n^3+1)}$. Let $\eps\in(0,1)$ and $\delta\in\ocint{0}{1}$, subject to $1-\eps dk<1$. Then there is an explicit $(\eps,(1-\delta)d)$-dimension expander of degree  $d\cdot \ceilp{\frac{dk}{\delta (1-\eps dk)}}$.
\end{corollary}
\begin{proof}
	As in \autoref{res:tensor-then-condense_instantiate}, we will instantiate \autoref{res:tensor-then-condense} with an explicit lossy condenser.  However, here we take the parameter $\gamma$ of \autoref{res:tensor-then-condense_instantiate} to have $\gamma=0$, as $\gamma>0$ was only needed to expand in the high-rate regime (which this proof will not achieve).  

	Thus, for $r=\floor{\eps n}$ we use \autoref{res:lossy-cond_small-field} to obtain an explicit $(\le rd,\delta)$-condenser in $\F_q^{n\times nd}$ with resulting seed length of $\ceilp{\frac{dk}{\delta(1-\eps dk)}}$ as
	\begin{align*}
		\ceilp{\frac{nd}{\delta(\floor{\nicefrac{n}{k}}-rd+1)}}
		\le\ceilp{\frac{nd}{\delta(\nicefrac{n}{k}-rd)}}
		=\ceilp{\frac{ndk}{\delta(n-rdk)}}
		\le\ceilp{\frac{dk}{\delta(1-\eps dk)}}
		\;,
	\end{align*}
	noting that we are restricted in needing that $\floor{\nicefrac{n}{k}}\cdot k\le n$ (and we pad matrices from $k\floor{\nicefrac{n}{k}}$ rows to $n$ rows). Plugging this into \autoref{res:tensor-then-condense} thus yields the desired dimension expander.
\end{proof}

Put more informally, this yields for every integer $d$, a $\left(\Theta\left(\frac{1}{d\log_q dn}\right),\Theta(d)\right)$-dimension expanders in $\F_q^n$ of degree $\Theta(d^2\log_q dn)$.

We now briefly remark on the difficulty of using more sophisticated code concatenation ideas from coding theory to obtain better results.

\begin{remark}\label{rmk:code-concate}
	The above constructions mimic the conversion of Reed-Solomon codes of block-length $n$ (over the large alphabet of $\F_{q^k}$, with $k=\ceil{\log_q n}$) to BCH codes of block-length $n$ (over the small alphabet of $\F_q$).  This conversion preserves the distance of the code, but multiplies the number of parity checks of the code by factor of $\ceil{\log_q n}$ for codes of block-length $n$.  As such, the relative distance of the resulting BCH code cannot be better than $\nicefrac{1}{\ceil{\log_q n}}$ without the code becoming trivial.

	Coding theory has the method of \emph{code concatenation} for reducing the alphabet size without incurring the above logarithmic loss.  In particular, one can reduce Reed-Solomon codes to a code over $\F_2$ while only incurring a constant factor loss in the distance and rate.  Unfortunately, it is unclear how to implement code concatenation in the context of this paper.  
	
	Specifically, in code concatenation for the Hamming metric, one first uses the isomorphism that $\F_q^{nk}=\F_{q^k}^n$, one then uses the outer-code $E_1:\F_{q^k}^n\to \F_{q^k}^m$, to which one then applies an inner-code $E_2:\F_{q^k}=\F_q^{k}\to\F_q^{k'}$ component-wise. In the Hamming metric, the isomorphism in the first step of the composition will preserve (or increase) the relative Hamming weight, so that a vector $\vv\in\F_q^{nk}$ with weight $\frac{d}{nk}$ will have weight at least $\frac{\nicefrac{d}{k}}{n}=\frac{d}{nk}$ when considered as a vector in $\F_{q^k}^n$. Thus, we have not lost in the distance by appealing to this isomorphism, which intuitively follows from the fact that the Hamming weight of a vector is defined coordinate-wise.
	
	However, this isomorphism does not seem to play well with rank (as the rank of a matrix is not defined column-wise), as discussed in the comments after \autoref{res:BCH-compose-rank}. That is, the rank can drop by a factor of $k$ when treating a $\F_q$-vector space $V\subseteq \F_q^{nk}$ as a $\F_{q^k}$-vector space in $\F_{q^k}^n$.  One can recover this factor of $k$ as done in the large-field simulation of Guruswami-Kopparty~\cite{GuruswamiKopparty13} in the context of subspace designs but this construction seems to lose the strong-ness of the construction which is essential for our work.
\end{remark}

\section{Constructions of Two-Source Rank Condensers}\label{sec:two-src_constr}

In this section, we relate two-source rank condensers to the pseudorandom objects we have already discussed.  First, we show that two-source rank condensers with good parameters (even for the special case when one of the sources has full rank) yield constructions of dimension expanders with good parameters.  We then show that \emph{seeded} (single-source) rank condensers can be used to construct two-source rank condensers with good parameters.  In particular, for $(r,r,0)$-condensers we obtain an output length of $\Theta(nr)$, which is essentially optimal. However, we note that this is essentially the construction of Forbes and Shpilka~\cite{ForbesShpilka12} for constructing \emph{rank-metric codes}.  We show in \autoref{sec:rank-metric} that (bilinear) lossless two-source rank condensers are \emph{equivalent} to (linear) rank-metric codes, and as a consequence derive \emph{optimal} such condensers from rank-metric code constructions.  For $(r,r,1-(1-\eps)^3)$-lossy two-source condensers, we show how to obtain output size $\Theta(\nicefrac{n}{\eps^2 r})$ for \emph{constant} $r$ by using our ideas as an ``outer condenser'' and finding via brute-force an ``inner condenser''.

We begin by showing how two-source rank condensers imply dimension expanders.

\begin{proposition}\label{res:two-src_dim-exp}
	Let $\F_q$ be a finite field. Let $n\ge r\ge 1$, $m\ge 1$, and $\eps>0$. Let $E_1,\ldots,E_t\in\F^{n\times m}$ be such that $f:\F^n\times\F^m\to\F^t$ with $f(\vv,\vw)\eqdef (\vv^\T E_i\vw)_{i\in[t]}$ is a bilinear $(\le r,m,\eps)$-two-source rank condenser.

	Then, for $i\in[m]$, define $A_i\in\F^{t\times n}$ to the be matrix of the linear transformation $\vv\mapsto f(\vv,\ve_i)$, where $\ve_i\in\F^m$ is the $i$-th standard basis vector.  Then $\cA\eqdef \{A_i\}_{i=1}^m$ has the following property: for all $V\subseteq\F^n$ of dimension $\le r$,
	\[
		\dim\sum_i A_i(V)\ge (1-\eps)m \dim V
		\;.
	\]

	In particular, consider $\delta\eqdef\nicefrac{r}{n}$, $\alpha\eqdef (1-\eps)m$ with $\alpha\delta<1$, where $m,\delta,\alpha$ and $\eps$ are constant.  For $n\ge\Omega_{m,\alpha,\delta,\eps}(1)$ whenever
	\[
			t\le	\frac{n}{\eps m}+\frac{m}{\eps}+(1-\eps)rm+o_q(1)+\O(1)
			\;,
	\]
	we have that the collection $\cA$ is a degree-$m$ $(\delta,\alpha)$-dimension expander in $\F^n$ whenever
	\[
		m>\alpha+\frac{1}{1-\alpha\delta}
		\;.
	\]
\end{proposition}
\begin{proof}
	First note that for each $i\in[m]$, the map $\vv\mapsto (\vv^\T E_j\ve_i)_{j\in[t]}$ is indeed a linear map in $\vv$ from $\F^n\to\F^t$, so that $A_i$ is well defined. Now consider a subspace $V\subseteq\F^n$ of dimension $\le r$.  By bilinearity of $f$, we have that $\dim f(V,\F^m)=\dim f(V,\Id_m)$, where $\Id_m$ is the $m\times m$ identity matrix so that $\cspn \Id_m=\F^m$. The rank condenser property guarantees that $\dim f(V,\Id_m)\ge (1-\eps)m\dim V$.  In particular, as $\{f(V,\ve_i)\}_{i\in[m]}$ spans $\sum_i A_i(V)$, it follows that $\dim\sum_i A_i(V)\ge (1-\eps)m \dim V$ as desired.

	To see the second part of the claim, note that for $\cA$ to be a dimension expander we only need $t\le n$ as we can then pad the matrices in $\cA$ to be $n\times n$ matrices. Thus it suffices that
	\[
			\frac{n}{\eps m}+\frac{m}{\eps}+(1-\eps)rm+o_q(1)+\O(1)\le n
			\;,
	\]
	for which it suffices that
	\[
			\frac{1}{\eps m}+o_n(1)+(1-\eps)\delta m+o_n(1)\le 1
			\;,
	\]
	which is equivalent to
	\[
		\frac{1}{\eps m}\le 1-\alpha\delta-o_n(1)
		\;.
	\]
	Using that $\eps m=m-\alpha$, we see this is equivalent to
	\[
		m\ge \alpha+\frac{1}{1-\alpha\delta-o_n(1)}
		\;,
	\]
	which will be satisfied for large $n$ as long as $m,\eps,\alpha,\delta$ are constants and $m>\alpha+\frac{1}{1-\alpha\delta}$.
\end{proof}

Thus, the above shows that obtaining two-source rank condensers meeting the probabilistic method bound of \autoref{res:two-src_prob-method<} will yield dimension expanders essentially meeting the probabilistic method bound for dimension expanders (\autoref{res:dim-exp_prob-method}).

In the definition of a two-source condenser, we say that it is bilinear if each coordinate of the output function is a bilinear form.  It will sometimes be more convenient to consider all of the coordinates together.  In this case, a bilinear condenser acts as a matrix $E$ times the tensor product $\vv\otimes\vw$ as we now show.

\begin{lemma}\label{res:bilinear-cond_alt}
	Let $\F$ be a field and let $f:\F^n\times\F^m\to\F^t$.  Let $E\in\F^{t\times nm}$ and for $i\in[t]$ define $E_i\eqdef E_{i,\bullet}\in\F^{n\times m}$, the $i$-th row of $E$ interpreted via the isomorphism $\F^{n\times m}=\F^{nm}$. Then $f$ has $f(\vv,\vw)=(\vv^\T E_i\vw)_{i=1}^t$ for $E_i\in\F^{n\times m}$ iff $f(\vv,\vw)=E\cdot (\vv\otimes\vw)$, where $\vv\otimes\vw\in\F^{nm}$ is the tensor product of $\vv$ and $\vw$.
\end{lemma}
\begin{proof}
	\begin{align*}
		(E\cdot (\vv\otimes\vw))_i
		&=\sum_{j\in[nm]} E_{i,j} \cdot(\vv\otimes\vw)_{j}
		\intertext{using the isomorphism between $\F^{nm}$ and $\F^{n\times m}$ so that we now index by $[n]\times[m]$,}
		&=\sum_{j\in[n],k\in[m]} (E_i)_{j,k} \cdot (\vv\otimes\vw)_{j,k}\\
		\intertext{factoring the tensor product,}
		&=\sum_{j\in[n],k\in[m]} (E_i)_{j,k} \cdot v_j\cdot w_k\\
		&=\sum_{j\in[n],k\in[m]} v_j (E_i)_{j,k} w_k\\
		&=\vv^\T E_i \vw
		\;.
		\qedhere
	\end{align*}
\end{proof}

As such, for bilinear rank condensers we can simply consider the matrix $E$ to be the condenser and not discuss the function $f$.

Similar to \autoref{res:subspace design_les-vs-eqs} we can obtain that lossless two-source condensers automatically work for \emph{all} smaller rank bounds. 

\begin{lemmawp}\label{res:two-src_lossless_le}
	Let $\F$ be a field and let $n\ge r\ge 1$ and $m\ge s\ge 1$. Let $f:\F^n\times\F^m\to\F^t$.  Then $f$ is a lossless $(r,s,0)$-two-source rank condenser iff $f$ is a lossless $(\le r,\le s,0)$-condenser.
\end{lemmawp}

However, just as in \autoref{res:dim-le-r_dim-eq-r} this is provably false for lossy condensers.  However, while for single-source condensers we can expect to (and do) obtain ``rank $\le r$'' results essentially for free from ``natural'' constructions obtaining results for ``rank $=r$'', this is provably not the case for two-source condensers.  In particular, we now show that condensing all small enough sources can induce a linear output lower bound, showing that the linear dependence in the output given by \autoref{res:two-src_prob-method<} for $(\le r,s,\eps)$-condensers is needed (in contrast to $(\le r,s,\eps)$-condensers which can do better).

\begin{proposition}\label{res:two-source_output-lb}
	Let $\F$ be a field and let $n\ge 1$ and $m\ge s\ge 1$. Let $f:\F^n\times\F^m\to\F^t$ where $f$ is a bilinear $(1,s,\eps)$-two-source rank condenser for $\eps<1$.  Then $t\ge m-\eps s$.
\end{proposition}
\begin{proof}
	The condenser is defined by a matrix $E\in\F^{t\times nm}$ such that for any $A\in\F^{n\times 1}$ of rank $1$ and $B\in\F^{m\times s}$ of rank $s$ we have that $\rank E(A\otimes B)\ge (1-\eps)1\cdot s$.  Thus, take $A=\{\ve_1\}$ where $\ve_1\in\F^n$ is the first standard basis vector.  Thus $E(A\otimes B)=E_1 B$, for some $E_1\in\F^{t\times m}$ as once we fix $\vv$, the map $\vw\mapsto E(\vv\otimes \vw)$ is linear. 
	
	Thus, we need that $E_1\in\F^{t\times m}$ has that for any $B\in\F^{m\times s}$ that $\rank E_1 B\ge (1-\eps)\rank B$.  If $\ker E_1$ has dimension $> \eps s$ then we can find a subspace $V$ with $\dim V=s$ and $\dim (V\cap \ker E_1)>\eps s$, so $\dim V-\dim E_1 V>\eps s$ and thus $E$ is not such a condenser.  Thus, it must be that $\dim \ker E_1\le \eps s$.  However, as $\dim \ker E_1\ge m-t$ it follows that $t\ge m-\eps s$.
\end{proof}

\subsection{Constructing Optimal Lossless Two-Source Condensers}

We now turn to constructing two-source rank condensers.  In this subsection, we consider the lossless case ($\eps=0$).  Thus, we are given $A\in\F^{n\times r}$ of rank $r$ and $B\in\F^{m\times s}$ and wish to obtain rank $rs$ from them.  One can trivially do this via the tensor product, so that $A\otimes B\in\F^{nm\times rs}$ has rank $rs$. However, we would like a condenser with a smaller output size.  To achieve this, we observe that we can apply seeded rank condensers $\cE:\F^n\to\F^r$ and $\cE':\F^m\to\F^s$ and after this dimension reduction apply the tensor product.  This will certainly yield rank $rs$ for \emph{some} $E\in\cE$ and $E'\in\cE'$, so we must enumerate over all such choices.  While this would naively yield an output size of $|\cE||\cE'|rs$, we can use that our lossless condensers (\autoref{res:lossless-extractor}) only have a finite number of bad values so that a union bound yields an output size of $(|\cE|+|\cE'|)rs$, as stated in the following proposition.

\begin{proposition}\label{res:two-src_condense-tensor}
	Let $\F$ be a field. Let $n\ge r\ge 1$ and $m\ge s\ge 1$. Let $\omega\in\F$ be an element of multiplicative order $\ge n,m$. Identify $\F^n$ and $\F^m$ with the spaces of low-degree univariate polynomials, so that $\F^n=\F[x]^{<n}$ and $\F^m=\F[y]^{<m}$ so that $\F^n\otimes\F^m=\F^{nm}=\F[x,y]^{<n,<m}$ is the space of bivariate functions with the respective individual degree bounds.  Let $S\subseteq\F\setminus\{0\}$ be a set where $|S|=r(n-r)+s(m-s)+1$.  Then
	$E:\F[x,y]^{<n,<m}\to\F^{rs\cdot|S|}$ defined by
	\[
		h(x,y)\mapsto ( h(\omega^i \alpha,\omega^j \alpha) )_{i\in\zr{r},j\in\zr{s},\alpha\in S}
	\]
	is a bilinear $(r,s,0)$-two-source rank condenser with output size $\le rs(rn+sm)$.
\end{proposition}
\begin{proof}
	By construction and \autoref{res:bilinear-cond_alt} we see that $E$ is bilinear and has the desired output size (as $r,s\ge 1$), so that it suffices to show the condensing property.

	Assume the setup of \autoref{constr:folded-wronskian}. In this construction, we see that $\Wr_r(\alpha)$ is a linear map $\Wr_r(\alpha):\F[x]^{<n}\to\F^r$ given by $f(\alpha)\mapsto (f(\alpha),f(\omega\alpha),\ldots,f(\omega^{r-1}\alpha))$, and $\Wr_s(\alpha):\F[y]^{<m}\to\F^s$ defined similarly (where we abuse the notation of \autoref{constr:folded-wronskian} so that we allow $\Wr_s(\alpha)$ to act on $\F[y]^{<m}$).  Thus, we see that $(\Wr_r(\alpha)\otimes \Wr_s(\alpha)):\F[x,y]^{<n,<m}\to\F^{rs}$ is a linear map sending $h(x,y)\mapsto ( h(\omega^i \alpha,\omega^j \alpha) )_{i\in\zr{r},j\in\zr{s}}$.

	Consider $A\in(\F[x]^{<n})^r$ of rank $r$ and $B\in(\F[y]^{<m})^s$ of rank $s$. It suffices to show that $\rank E(A\otimes B)\ge rs$ and as $E$ is simply the collection of the maps $(\Wr_r(\alpha)\otimes \Wr_s(\alpha))$ for $\alpha\in S$, it suffices to show that $\rank \bigl( (\Wr_r(\alpha)\otimes \Wr_s(\alpha)) \cdot (A\otimes B) \bigr) =\rank ((\Wr_r(\alpha)A)\otimes (\Wr_s(\alpha) B))$ has rank $\ge rs$ for some $\alpha\in S$.
	
	It follows from \autoref{res:folded-wronskian-nonzero} that there are at most $r(n-r)$ values of $\alpha$ where $\rank \Wr_r(\alpha)A<\rank A$.  Similarly, there are at most $s(m-s)$ values of $\alpha$ where $\rank \Wr_s(\alpha)B<\rank B$.  Thus, as $|S|> r(n-r)+s(m-s)$ it follows there is some $\alpha_0\in S$ where $\rank \Wr_r(\alpha_0)A=\rank A$ and $\rank \Wr_s(\alpha_0)B=\rank B$.  Thus, it follows that $\rank ((\Wr_r(\alpha_0)A)\otimes (\Wr_s(\alpha_0)B))=rs$ as desired.
\end{proof}

Note that in the balanced case of $n=m$ and $r=s$, this yields an output size of $\approx 2nr^3$ which while better than the trivial $n^2$ is still far from the probabilistic method bound of $2nr$ (\autoref{res:two-src_prob-method}).  However, we observe that the above map has a fair bit of \emph{redundancy} which we can prune, and we can prune this redundancy because of the following lemma.

\begin{lemma}\label{res:two-src_prune}
	Let $\F$ be a field. Let $n,m\ge 1$, $1\le r\le n$ and $1\le s\le m$.  Let $E\in\F^{t\times nm}$ be a bilinear $(r,s,\eps)$-two-source rank condenser and suppose that $E'\in\F^{t'\times nm}$ has $\rspn E\subseteq\rspn E'$.  Then $E'$ is a bilinear $(r,s,\eps)$-two-source rank condenser.
\end{lemma}
\begin{proof}
	That $\rspn E\subseteq\rspn E'$ means that there is a $P\in\F^{t\times t'}$ so that $E=PE'$.  Now consider some $A\in\F^{n\times r}$ of rank $r$ and $B\in\F^{m\times s}$ of rank $s$. That $E$ is such a condenser means that $\rank E(A\otimes B)\ge (1-\eps)rs$.  But then $\rank E(A\otimes B)=\rank PE'(A\otimes B)\le \rank E'(A\otimes B)$.  Thus it follows that $\rank E'(A\otimes B)\ge (1-\eps)rs$, so that $E'$ is also such a condenser.
\end{proof}

We now use the above lemma to prune the redundancies of \autoref{res:two-src_condense-tensor} to obtain a near-optimal lossless two-source condenser.

\begin{corollary}\label{res:two-src_condense-tensor_pruned}
	Let $\F$ be a field. Let $n\ge r\ge 1$ and $m\ge s\ge 1$. Let $\omega\in\F$ be an element of multiplicative order $\ge n,m$. Identify $\F^n$ and $\F^m$ with the spaces of low-degree univariate polynomials, so that $\F^n=\F[x]^{<n}$ and $\F^m=\F[y]^{<m}$ so that $\F^n\otimes\F^m=\F^{nm}=\F[x,y]^{<n,<m}$ is the space of bivariate functions with the respective individual degree bounds.  Let $T\subseteq\F\setminus\{0\}$ be a set where $|T|=n+m-1$.  Then
	$E':\F[x,y]^{<n,<m}\to\F^{(r+s-1)\cdot|T|}$ defined by
	\[
		h'(x,y)\mapsto ( h'(\beta,\omega^k \beta) )_{-r<k<s,\beta\in T}
	\]
	is a bilinear $(r,s,0)$-two-source rank condenser with output size $(r+s-1)(n+m-1)$.
\end{corollary}
\begin{proof}
	Note that the map of \autoref{res:two-src_condense-tensor} consists of evaluating $h(x,y)\in \F[x,y]^{<n,<m}$ at the points  $h(\omega^i \alpha,\omega^j \alpha)$ for $i\in\zr{r}$, $j\in\zr{s}$, and $\alpha$ in some set.  Note that each such evaluation is the evaluation of the corresponding \emph{univariate} polynomial $h(x,\omega^{j-i}x)$ at the point $\beta=\omega^i\alpha$.  As the polynomial $h(x,\omega^{j-i}x)$ is of degree $\le (n-1)+(m-1)$ and there are $(r-1)+(s-1)+1$ values of $k=i-j$, we see that by polynomial interpolation that the evaluations of the map \autoref{res:two-src_condense-tensor} are contained in the linear combinations of the evaluations
	\[
		h'(x,y)\mapsto (h'(\beta,\omega^k\beta))_{-r<k<s,\beta\in T}
		\;,
	\]
	since $|T|\ge(n-1)+(m-1)+1=n+m-1$.  Appealing to \autoref{res:two-src_prune} then yields the claim.
\end{proof}

Thus, in the balanced case of $n=m$ and $r=s$ we obtain an output of size $\le 4nr$, which matches the probabilistic method up to a factor of 2.  While this seems to close the story on lossless two-source rank condensers, we note here that this construction is \emph{exactly} the construction of Forbes and Shpilka~\cite{ForbesShpilka12} for an object called a \emph{rank metric code}.  While a priori it might seem that rank-metric codes are weaker objects (so that the above analysis would add something new), we show in \autoref{sec:rank-metric} the rank metric codes are \emph{equivalent} to lossless two-source rank condensers.  

\begin{propositionwp*}[\autoref{res:lossless-two-src_to_rank-metric}, \autoref{res:rank-metric_to_lossless-two-src}]
	Let $\F$ be a field. Let $n,m\ge 1$ and $n,m\ge r\ge0$. A matrix $E\in\F^{t\times nm}$ is a bilinear $(r,r,0)$-two-source rank condenser iff $\cC\eqdef \ker E\subseteq\F^{n\times m}$ is a distance $\ge r+1$ rank-metric code.
\end{propositionwp*}

Thus, we can then leverage constructions of rank-metric codes to slightly improve upon the above condenser.  In particular, Gabidulin~\cite{Gabidulin85a} rank-metric codes work over any fixed finite field, and the codes of Roth~\cite{Roth91} and Forbes-Shpilka~\cite{ForbesShpilka12} improve upon the above condensers by removing further redundancy.  Perhaps more importantly, these codes are \emph{optimal} for their respective regimes. In particular, we can translate the known bounds for rank-metric codes into the language of condensers.

\begin{propositionwp*}[\autoref{res:rank-metric_singleton}, \autoref{res:rank-metric_limit_roth}]
	Let $\F$ be a field and $m\ge n\ge r\ge 1$.  Let $f:\F^n\times\F^m\to\F^t$ be a bilinear $(r,r,0)$-two source rank condenser.  Then $t\ge rm$.  Further, if $\F$ is algebraically closed and $n=m$ then $t\ge r(2n-r)$.
\end{propositionwp*}

We can then match these bounds with constructions of rank-metric codes, appealing to the above equivalence with condensers.

\begin{propositionwp}[\autoref{res:rank-metric_code_gabidulin}, \autoref{res:rank-metric_code_roth}]\label{res:two-src_constr}
	Let $\F$ be a field and $m\ge n\ge r\ge 1$. 
	
	If $\F=\F_q$ is a finite field, and $\F_{q^m}$ is given explicitly as an extension of $\F_q$ then there is an explicit $f:\F^n\times\F^m\to\F^t$ which is a bilinear $(r,r,0)$-two source rank condenser with $t=rm$.

	If $|\F|\ge n$ then there is another such explicit $f$ with $t=r(n+m-r)$.
\end{propositionwp}

We note that the last part of the above result can be proven by a tighter analysis of the construction used in \autoref{res:two-src_condense-tensor_pruned}, but the cleanest exposition goes through rank-metric codes.

\subsection{Constructing Lossy Condensers}\label{sec:two-src_lossy}

We now turn to constructions of lossy two-source condensers, where our results are still far from optimal. Because of this distance to optimality, we restrict our attention to balanced sources, so that we seek $(r,r,\eps)$-condensers $f:\F^n\times\F^n\to\F^t$. We begin by discussing how the strategy in the lossless case (condense each source with a seeded condenser, tensor the results, then finally prune the output) does not seem to give any better results.  Despite this, we show that this idea can serve well as a ``outer condenser'', so that if we had a good ``inner condenser'' for rank $r$ subspaces of an $r^3$-dimensional space we would obtain near-optimal results.

Recall the strategy of the previous section.  That is, we obtained lossless two-source condensers by using a seeded condenser based on the folded Wronskian (\autoref{constr:folded-wronskian}), where we map $f(x)\mapsto (f(\alpha),f(\omega\alpha),\ldots,f(\omega^{t-1}\alpha))$ for various $\alpha$.  While using this approach yielded suboptimal results when applied naively (\autoref{res:two-src_condense-tensor}) we saw how to prune the output to yield an essentially optimal result (\autoref{res:two-src_condense-tensor_pruned}).  The key idea to this pruning was that the induced map 
\[
	f(x,y)\mapsto (f(\omega^i\alpha,\omega^j\alpha))_{i,j\in\zr{t},\alpha\in S}
\]
evaluates $f$ through evaluations of \emph{univariate} polynomials $f(x,\omega^k x)$, and we can then restrict the number of evaluations to each univariate polynomial to be at most the degree bound of that polynomial.  This helps as $k\cdot |S|$ was much larger than the degree of $f(x,\omega^k x)$.  

One can thus implement the above strategy in the lossy case, where we now use the analysis of the folded Wronskian as a lossy rank condenser (\autoref{res:gk13-lossy-condenser}).  Thus, before any attempt at pruning we see that the output size is $t^2\cdot \frac{n}{\eps(t-r)}$, which is minimized at $\approx nr$ for $t=2r$. While this improves upon the non-pruned lossless condenser of \autoref{res:two-src_condense-tensor} it does not even yield a smaller output size than the \emph{lossless} constructions we give above.

Thus, to obtain better results we might try to prune this construction.  However, we see that each univariate $f(x,\omega^k x)$ is evaluated only at $\frac{n}{\eps(t-r)}$ points, which is \emph{below} the degree bound of $\deg f(x,\omega^k x)=2n$.  Thus, there seems to be no pruning available to improve the above strategy.

Despite this, we observe that the above strategy has now reduced the problem to a smaller dimensional space, for which one could apply different methods.  We now show that taking $t=r^3$ in this reduction can lead to near-optimal results.

\begin{proposition}\label{res:two-src_lossy-outer}
	Let $\F$ be a field and $n\ge r\ge 1$. Assume the setup of \autoref{constr:folded-wronskian}.  Suppose that $E\in\F^{t'\times t^2}$ is a bilinear $(\ceil{(1-\eps)r},\ceil{(1-\eps)r},\eps)$-two source rank condenser.  
	Let $S\subseteq\{(\omega^t)^j\st j\ge 0\}$ with $|S|=\ceilp{\frac{2n}{\eps(t-r+1)}}$.  Then $E'\eqdef(E\cdot (\Wr_{t}(\alpha)\otimes \Wr_{t}(\alpha)))_{\alpha\in S}\in \F^{t'\cdot |S|\times n^2}$ is $(r,r,1-(1-\eps)^3)$-condenser.

	In particular, taking $t=r^3$ and taking $E$ to be a condenser meeting the bound of \autoref{res:two-src_prob-method} so that $t'\le \frac{2r^3}{\eps r}+(1-\eps)r^2+\O(1)$, we obtain that $E'$ is a $(r,r,1-(1-\eps)^3)$-condenser in $\F^{t''\times n^2}$ for $t''\le \O(\nicefrac{n}{\eps^2 r})$.
\end{proposition}
\begin{proof}
	Given $A,B\in\F^{n\times r}$ of rank $r$, we see that by \autoref{res:gk13-lossy-condenser} there are $<2\frac{n}{\eps(t-r+1)}$ values of $\alpha\in S$ such that $\rank \Wr_{t}(\alpha) A<(1-\eps)\rank A$ or $\rank \Wr_{t}(\alpha) B<(1-\eps)\rank B$.  Thus, as $|S|\ge \frac{2n}{\eps(t-r+1)}$ there is some $\alpha_0$ where $\rank \Wr_{t}(\alpha_0) A,\rank \Wr_{t}(\alpha_0) B\ge (1-\eps)r$.  Thus, it follows that as $E$ is a $(\ceil{(1-\eps)r},\ceil{(1-\eps)r},\eps)$-condenser that $E((\Wr_{t}(\alpha_0) A)\otimes (\Wr_{t}(\alpha_0) B))$ has rank $\ge (1-\eps)\cdot (1-\eps)^2r^2$.  Thus, it follows that $(E\cdot (\Wr_{t}(\alpha)\otimes \Wr_{t}(\alpha)))_{\alpha\in S}$ applied to $A\otimes B$ has rank $\ge (1-\eps)^3r^2$, showing that this is indeed the desired condenser.

	To obtain the bound on $t'$, note that
	\begin{align*}
		t''
		&=|S|\cdot t'\\
		&\le \ceilp{\frac{2n}{\eps(r^3-r+1)}}\cdot \left(\frac{2r^3}{\eps r}+(1-\eps)r^2+\O(1)\right)\\
		&=\O\left(\frac{n}{\eps^2 r}\right)
		\;.
		\qedhere
	\end{align*}
\end{proof}

In particular, one could find such a condenser $E$ via brute force when $r=\O(1)$.

\section{Open Questions}

This work leaves several directions for future work.

\begin{enumerate}
	\item Can one obtain $(r,\eps)$-lossy seeded rank \emph{extractors}, where the output is $\approx (1-\eps)r$? Our methods require the output to be $\ge r$.

	\item Can one develop of theory of ``code concatenation'' to improve our results in \autoref{sec:small-fields} for small fields?

	\item Can one obtain lossy two-source rank condensers with output size $o(nr)$ for $r=\omega(1)$?

	\item Can one obtain \emph{lossless} dimension expanders, where the degree/expansion relationship matches the probabilistic method?

	\item What is the complexity of computing dimension expansion? That is, given matrices $A_1,\ldots,A_d\in\F^{n\times n}$, compute the largest $\alpha$ so that $\cA\eqdef\{A_i\}_{i=1}^d$ is a $(\nicefrac{1}{2},\alpha)$-dimension expander.
\end{enumerate}

\section*{Acknowledgments}

We would like to thank 
Swastik Kopparty,
Prasad Raghavendra,
Amir Shpilka,
Amir Yehudayoff,
and
Avi Wigderson
for helpful comments.

\bibliographystyle{mfalphaurl}
{
	\bibliography{two-src-rk-cond}
}

\appendix

\section{Toward Iterative Constructions of Subspace Evasive Sets}\label{sec:subspace-evasive}

We describe here a motivation for studying two-source rank condensers via a potential application to constructing another object called a \emph{subspace evasive set}, as defined by Guruswami~\cite{Guruswami11}.

\begin{definition}[Guruswami~\cite{Guruswami11}]
	A set $S\subseteq\F^n$ is a \textbf{$(r,L)$-subspace evasive set} if for every $(\le r)$-dimensional subspace $V\subseteq\F^n$, $|V\cap S|\le L$. Equivalently, a set $S$ is $(r,L)$-subspace evasive if for every subset $T\subseteq S$ with $|T|=L+1$, $\rank T\ge r+1$.
\end{definition}

This notion was introduced as a method to prune the lists in list-decodable codes when using a linear algebraic approach that pins down candidate messages to a low-dimensional subspace~\cite{Vadhan12,Guruswami11,GuruswamiWang13}. To obtain a high rate code after this pruning we desire a \emph{large} subspace evasive set.  Such large sets are guaranteed by the probabilistic method, as given in the following lemma.

\begin{lemmawp}[Guruswami~\cite{Guruswami11}]
	Let $\F_q$ be a finite field. Let $\eps>0$ and $n\ge r\ge 1$. There exists a $(r,\nicefrac{2r}{\eps})$-subspace evasive set in $\F_q^n$ of size $\floor{q^{(1-\eps)n}}$ when $r\le\nicefrac{\eps n}{2}$.
\end{lemmawp}

The work of Ben-Aroya and Shinkar~\cite{BenAroyaShinkar12} showed that for $r\ge\left(\nicefrac{1}{\eps}\right)^{\Omega(1)}$ the above size-bound of ``$\nicefrac{2r}{\eps}$'' is asymptotically optimal up to constants.

Note that the most basic construction of a subspace evasive set is based on taking the set of moment curve vectors $S\eqdef\{(1,\alpha,\ldots,\alpha^{n-1})\st \alpha\in\F_q\}$.  This set $S$ is highly evasive (as it is $(n,n-1)$-subspace-evasive), but has size $q$, which is quite far from the above $\floor{q^{(1-\eps)n}}$.  Thus far, the best explicit constructions are due to Dvir-Lovett~\cite{DvirLovett12} and Ben-Aroya and Shinkar~\cite{BenAroyaShinkar12}.

\begin{theoremwp*}[Dvir and Lovett~\cite{DvirLovett12}]
	Let $\F_q$ be a finite field. Let $\eps>0$ and $n\ge r\ge 1$. Then there exists an explicit $(r,\left(\nicefrac{r}{\eps}\right)^r)$-subspace evasive set in $\F_q^n$ of size $>q^{(1-\eps)n}$.
\end{theoremwp*}

\begin{theoremwp*}[Ben-Aroya and Shinkar~\cite{BenAroyaShinkar12}]
	Let $\F_q$ be a finite field. Let $0<\eps<\nicefrac{1}{2}$ and $n\ge r\ge 1$ where $r,\nicefrac{1}{\eps}\le\O(1)$. Then there exists an explicit $(r,\left(\nicefrac{2}{\eps}\right)^r)$-subspace evasive set in $\F_q^n$ of size $\ge q^{(1-\eps)n}$.
\end{theoremwp*}

Note that while the latter result only works for constant $r$ and $\eps$, this is still an interesting parameter regime for list-decoding applications.

Given that we are far from optimal explicit constructions for subspace evasive sets, we consider in this work an iterative approach to constructing subspace evasive sets along the lines of the zig-zag product of Reingold, Vadhan and Wigderson~\cite{ReingoldVadhanWigderson02} for constructing expanders.  That is, suppose that $S\subseteq\F_q^n$ is $(r,L)$-subspace evasive.  It follows that $S\otimes S=\{\vv\otimes\vw\st\vv,\vw\in S\}$ is in a certain sense subspace-evasive.  That is, for all $T,R\subseteq S$ with $|T|,|R|=L+1$, the product set $T\otimes R\subseteq S\otimes S$ has dimension $\ge (r+1)^2$.  Thus, large enough ``product sets'' among $S\otimes S$ have high rank. While having this property for all product sets is far from having it for all sets of similar size, and thus this is far from a iterative approach for boosting evasiveness, there is also the more basic problem that the \emph{rate} of this construction is poor.  
That is, ideally we have that the rate $\frac{\log_q|S|}{n}$ is at least $\Omega(1)$.  However, $\frac{\log_q|S\otimes S|}{n^2}=\frac{2}{n}\frac{\log_q|S|}{n}$ and thus this tensoring operation decreases the rate dramatically.

While the ``product versus non-product set'' problem seems fundamental to the above approach, we try to ameliorate the rate issue by ``derandomizing'' the tensor product using a two-source rank condenser.
That is, if $S$ is $(r,L)$-subspace evasive and $f:\F^n\times \F^n\to\F^t$ is a $(r+1,r+1,\nicefrac{1}{2})$-condenser, then we can define $S\otimes_f S\eqdef\{ f(\vv,\vw)\st \vv,\vw\in S\}$ and observe that for any subset $R,T\subseteq S$ with $|R|,|T|\ge L+1$ we have that $R\otimes_f T\subseteq S\otimes_f S$ has rank $\ge \nicefrac{(r+1)^2}{2}$.  While this evades slightly smaller subspaces than using the standard tensor product, this operation still qualitatively squares the evasiveness as long as $r\ge \Omega(1)$ initially.  Further, the rate has not dropped substantially if good enough condensers are used.  That is, the probabilistic method (\autoref{res:two-src_prob-method}) shows that we can take $t=\Theta(\nicefrac{n}{r}+r^2)$.  Thus, as long as $r\le\O(\sqrt[3]{n})$, we get that $t=\Theta(\nicefrac{n}{r})$ and that $\frac{\log_q|S\otimes_f S|}{t}=\Theta(2r\frac{\log_q|S|}{n})$.  Thus, as long as $r\ge\Omega(1)$ the rate has actually \emph{increased} under this derandomized tensor product.  Thus, from the perspective of iteratively constructing subspace evasive sets such condensers allow one to focus on the ``product versus non-product set'' problem.

Unfortunately in this work we do not construct such condensers with $t \approx n/r$, nor are we even able to obtain $t=\O(n)$ for any $r\ge\omega(1)$.  

\section{Lossless Two-source Condensers versus Rank-Metric Codes}\label{sec:rank-metric}

In this section we define the notion of a \emph{rank metric code} and show that it is equivalent to our notion of a two-source condenser \autoref{defn:two-source}.  We then review constructions and limitations of rank-metric codes.  Finally we derive the corresponding results in the language of two-source condensers.

\subsection{Equivalence of Condensers and Rank-Metric Codes}

We begin with the definition of rank metric codes.

\begin{definition}\label{defn:rank-metric}
	Let $\F$ be a field and $n,m\ge r\ge 1$. A set $\cC\subseteq\F^{n\times m}$ is a \demph{rank metric code with distance $r$} if for all $A,B\in\cC$ with $A\ne B$, $\rank(A-B)\ge r$.  The code is \demph{linear} if $\cC$ is a linear space. The \demph{parity checks} of a linear code $\cC$ consist of a basis for the dual space $\cC^\perp$.  A linear rank metric code $\cC$ is \demph{explicit} if one can construct a basis of $\cC$ in $\poly(n,m)$ operations in $\F$.
\end{definition}

When the linear code $\cC$ is simply the zero subspace we define the distance to be $\min\{n,m\}+1$. We now observe that via standard coding theory arguments linear codes have their distance equal to the minimum weight of a non-zero element in the code (where ``weight'' here means ``rank'').

\begin{lemmawp}
	Let $\F$ be a field and $n,m\ge 1$. Let $\cC\subseteq\F^{n\times m}$ be a linear rank-metric code.  Then the distance of $\cC$ is $\min_{0\ne M\in\cC}\rank M$.
\end{lemmawp}

We now give a lemma relating the inner product of matrices to the trace of their product, which will be helpful in the below results.

\begin{lemma}\label{res:matrix-inner-product}
	Let $\F$ be a field and $n, m\ge1$.  Let $M,N\in\F^{n\times m}$.  Then the inner product of $M$ and $N$ as vectors in $\F^{nm}$ can be expressed as $\la M,N\ra=\tr(M^\T N)=\tr(N^\T M)$.
\end{lemma}
\begin{proof}
	It suffices to prove the first part, as $\la M,N\ra=\la N,M\ra$.
	\begin{equation*}
		\la M,N\ra
		=\sum_{i\in[n],j\in[m]} M_{i,j} N_{i,j}
		=\sum_{i\in[n],j\in[m]} (M^\T)_{j,i} N_{i,j}
		=\sum_{j\in[m]} ((M^\T)\cdot N)_{j,j}
		=\tr(M^\T N)
		\;.
		\qedhere
	\end{equation*}
\end{proof}

We now show that a matrix $E\in\F^{t\times nm}$ defining a bilinear lossless two-source condenser has a kernel which is a good rank-metric code. In particular, if $E$ had a low-rank matrix $M\in\F^{n\times m}$ in its kernel, than $E$ would not losslessly preserve the row-span and column-span of $M$.

\begin{proposition}\label{res:lossless-two-src_to_rank-metric}
	Let $\F$ be a field and let $n,m\ge 1$ and $n,m\ge r\ge 0$. Let $f:\F^n\times\F^m\to\F^t$ be a bilinear lossless $(r,r,0)$-two-source rank condenser defined by $f(\vv,\vw)=(\vv^\T E_k\vw)_{k\in[t]}$ with $E_k\in\F^{n\times m}$.  Then $\{E_k\}_k$ are the parity checks of a linear rank-metric code with distance $\ge r+1$.
\end{proposition}
\begin{proof}
	We show the contrapositive for $r>0$, as when $r=0$ the claim is trivial as every rank-metric code has distance $\ge 1$.  That $\{E_k\}_k$ are not such parity checks means that there is a non-zero matrix $M\in\F^{n\times m}$ with rank $s \le r$ such that $(\la E_k,M\ra)_{k\in[t]}=\vec{0}$. Let $M=AB^\T$ with $A\in\F^{n\times s}$ and $B\in\F^{m\times s}$. 
	
	As $(f(\vv,\vw))_k=\vv^\T E_k\vw$, it follows that we can express the $k$-th output of $f(A,B)$ as $(f(A,B))_k=A^\T E_k B\in\F^{s\times s}$.  Now consider the inner product of these vectors with the $s\times s$ identity matrix $\Id_s\in\F^{s\times s}$, and appealing to \autoref{res:matrix-inner-product} we have that $\la \Id_s,A^\T E_k B\ra=\tr(\Id_s A^\T E_k B)=\tr(BA^\T E_k)=\tr(M^\T E_k)=\la E_k,M\ra$, using the fact that the trace is invariant under cyclic permutations.
	
	As $(\la E_k,M\ra)_{k\in[t]}=\vec{0}$ it follows that $\la \Id_s,A^\T E_k B\ra=0$ for all $k$. That is, the $t\times s^2$ matrix which is the output $f(A,B)$ has a non-trivial nullspace as it contains the matrix $\Id_s$.  Thus, $\rank f(A,B)<s^2=\rank A\cdot \rank B$, so that $f$ is not a lossless rank condenser.
\end{proof}

Somewhat surprisingly, we also show the converse to the above, showing that the parity checks of any rank-metric code yield a bilinear lossless two-source rank condenser. This follows from the observation that if a matrix $E\in\F^{t\times nm}$ fails to condense the pair of subspaces $A\in\F^{n\times r}$ and $B\in\F^{m\times s}$ then $E(A\otimes B)$ must have rank $<rs$ so that it has a linear dependence $E(A\otimes B)C=0$ where $C\in\F^{rs\times 1}$.  However, using that $\F^{rs\times 1}=\F^{r\times s}$, we can see that $(A\otimes B)C$ can be interpreted as the matrix $ACB^\T \in \F^{n \times m}$, which is of rank $\le r,s$.  Thus, $E$ has a low-rank matrix in its kernel, and so the kernel does not define a good rank-metric code.

\begin{proposition}\label{res:rank-metric_to_lossless-two-src}
	Let $\F$ be a field and let $n,m\ge 1$ and $n,m\ge d\ge 0$. Let $E_1,\ldots,E_t\in \F^{n\times m}$ be the parity check matrix of a linear rank-metric code with distance $d+1$.  Define $f:\F^n\times \F^m\to\F^t$ by $(f(\vv,\vw))_k\eqdef \vv^\T E_k \vw$.  Then $f$ is a bilinear lossless $(d,m,0)$-two-source rank condenser, as well as a bilinear lossless $(n,d,0)$-two-source rank condenser.
\end{proposition}
\begin{proof}
	We consider the contrapositive.  Suppose that $A\in\F^{n\times r}$ with $\rank A=r$ and $B\in\F^{m\times s}$ with $\rank B=s$ is not losslessly condensed, so that $\rank f(A,B) <rs$. The output of $f(A,B)$ is a set of $rs$ vectors, each in $\F^t$, which we can consider as a $t\times rs$ matrix.  Viewing this matrix row-wise we can express the $k$-th output of $f(A,B)$ as $(f(A,B))_k=A^\T E_k B$.  Thus, that $\rank f(A,B) <rs$ means that there is some non-zero matrix $C\in\F^{r\times s}$ such that $\la C,(f(A,B))_k\ra=0$ for all $k$.  Thus, (appealing to \autoref{res:matrix-inner-product})
	\[
		0
		=\la C,(f(A,B))_k\ra
		=\tr(C^\T A^\T E_k B)
		=\tr(B C^\T A^\T E_k)
		=\la E_k, ACB^\T\ra
		\;,
	\]
	for all $k\in[t]$.  
	
	Note that we can decompose $C$ into the basis $\{\ve_i \ve_j^\T\}_{i\in[r],j\in[s]}$ for $\F^{r\times s}$ where $\ve_i$ is the $i$-th standard basis vector.  Thus,
	\[
		 ACB^\T
		 =A \left(\sum_{i\in [r],j\in[s]} C_{i,j}\ve_i\ve_j^\T\right) B^\T
		 =\sum_{i\in [r],j\in[s]} C_{i,j}\cdot A (\ve_i\ve_j^\T) B^\T
		 =\sum_{i\in [r],j\in[s]} C_{i,j}\cdot A_{\bullet,i} (B_{\bullet,j})^\T
		 \;,
	\]
	where $A_{\bullet,i}$ is the $i$-th column of $A$ and likewise for $B$. Now note that the outer-products $\{A_{\bullet,i} (B_{\bullet,j})^\T\}_{i\in[r],j\in[s]}\subseteq\F^{n\times m}$ are linearly independent as this is simply saying that the tensor product multiplies rank.  Thus, as $C\ne 0$, it follows that $ACB^\T$ is a non-trivial linear combination of the $\{A_{\bullet,i} (B_{\bullet,j})^\T\}_{i,j}$ and thus $ACB^\T\ne 0$.
	
	However, now note that $ACB^\T\in\F^{n\times m}$ has $\rank ACB^\T\le \min\{\rank A,\rank B\}$ so that $\rank ACB^\T\le \min\{r,s\}$.  Thus, if $\min\{r,s\}\le d$ then this shows that the $\{E_k\}_k$ are not the parity checks of a distance $d+1$ rank metric as $ACB^\T$ is a non-zero matrix with rank $<d+1$ where $(\la E_k, ACB^\T\ra)_{k\in[t]}=\vec{0}$. Thus, the claim follows by taking first $r\le d$ and $s\le m$, then taking $r\le n$ and $s\le d$.
\end{proof}

It thus follows that focusing on balanced sources in lossless two-source condensers is sufficient, using that lossless condensers also condense smaller sources (\autoref{res:two-src_lossless_le}).

\begin{corollarywp}\label{res:two-src_rs-r}
	Let $\F$ be a field and let $n,m\ge r\ge 1$. Let $f:\F^n\times\F^m\to\F^t$ be a bilinear lossless $(r,r,0)$-two-source rank condenser. Then $f$ is also a $(\le r,\le m,0)$- and $(\le n,\le r,0)$-condenser.
\end{corollarywp}

\subsection{Constructions of Rank-Metric Codes}

We now review for completeness known constructions of rank-metric codes, which by the above equivalence yields constructions of bilinear lossless rank condensers.  In this section, we will assume that $m\ge n$, so that we work with short and fat matrices. We begin with a basic limitation of rank-metric codes which is a direct generalization of the Singleton bound for codes in the Hamming metric.

\begin{propositionwp}[Gabidulin~\cite{Gabidulin85a}, Delsarte~\cite{Delsarte78}, and Roth~\cite{Roth91}]\label{res:rank-metric_singleton}
	Let $\F$ be a field and $m\ge n\ge 1$ and $n\ge r\ge 0$.  Let $\cC\subseteq\F^{n\times m}$ be a rank-metric code with distance $\ge r+1$.  If $\cC$ is a linear code then $\dim \cC\le m(n-r)$.  If $\F$ is finite (and $\cC$ possibly non-linear) then $\log_{|\F|}|\cC|\le m(n-r)$.
\end{propositionwp}

We now give a well-known construction of rank-metric codes, now called \emph{Gabidulin codes}, which are a $q$-linearized version of Reed-Solomon codes.  Correspondingly, these codes meet the above analogue of the Singleton bound.

\begin{proposition}[Gabidulin~\cite{Gabidulin85a}, Delsarte~\cite{Delsarte78}, and Roth~\cite{Roth91}]\label{res:rank-metric_code_gabidulin}
	Let $\F_q$ be a finite field and $m\ge n\ge 1$ and $n\ge r\ge 0$. Given a presentation of $\F_{q^m}$ as a extension field $\F_q$, there is an explicit linear rank-metric code $\cC\subseteq\F^{n\times m}$ with distance $r+1$ with $\dim \cC=m(n-r)$.
\end{proposition}
\begin{proof}
	\uline{$n=m$:}  First, we call a polynomial $f(x)\in\F_{q^m}[x]$ to be \emph{$q$-linearized} if it can written as $f(x)=\sum_{i=0}^k \beta_i x^{q^i}$. Now note that the evaluation map $f:\F_{q^m}\to\F_{q^m}$ is $\F_q$-linear by the Frobenius endomorphism, thus we can consider $f$ as a linear map $M_f:\F_q^m\to\F_q^m$.  Further, we can treat $M_f$ as a matrix $M_f\in\F_q^{m\times m}$ by choosing some $\F_q$-basis $\alpha_1,\ldots,\alpha_m$ for $\F_{q^m}$ so that $\sum_i \gamma_i \alpha_i \to \sum_i \gamma_i f(\alpha_i)$ for $\gamma_i\in\F_q$.
	
	Now note that the roots of a $q$-linearized polynomial, by the above linearity, form a $\F_q$-linear space within $\F_{q^m}$.  In particular, if $f$ is $q$-linearized and has $\deg f\le q^k$ then it has $\le q^k$ roots, so the map $M_f$ has a kernel of $\F_q$-dimension $\le k$, so that $\rank M_f\ge m-k$.

	Finally, define $\cC_k\eqdef\{M_f\st f\in \F_{q^m}[x], f=\sum_{i=0}^k \beta_i x^{q^i}\}\subseteq\F_q^{m\times m}$ for $k<m$.  Noting that the degree $\le q^k$ $q$-linearized polynomials form a $\F_q$-vector-space, it follows that $\cC_k$ is a linear space of dimension $\le (k+1)m$.  Further, $\cC_k$ is a linear space of dimension $\ge (k+1)m$ as the map from polynomials to matrices is injective.  That is, a degree $\le q^k$ $q$-linearized polynomial that yields a zero matrix must be zero as the zero matrix has a kernel of size $q^n$ but a non-zero such polynomial has at most $q^k$ roots. Thus, $\cC_k$ is a rank metric code of distance $m-k$ and dimension $(k+1)m$. Taking $k=m-(r+1)<m$ yields the claim (for $r=m$ we get $k=-1$ so that $\cC_{-1}$ is just the zero polynomial).

	\uline{$n<m$:}  Consider the above code $\cC_{m-(r+1)}\subseteq\F^{m\times m}$ with distance $r+1$.  Now consider the subcode $\cC'$ of the above code where we insist that the last $m-n$ rows are all zero.  It follows that $\cC'$ is still a linear rank-metric code with distance $r+1$, and we can now embed it into $\F^{n\times m}$.  The dimension is at least $\dim\cC'\ge \dim \cC_{m-(r+1)}-m(m-n)=m(m-r)-m(m-n)=m(n-r)$. By the above Singleton-type bound (\autoref{res:rank-metric_singleton}) it follows that this lower bound is met with equality.

	\uline{explicitness:} This is clear from construction as the presentation of $\F_{q^m}$ yields a $\F_q$-basis for $\F_{q^m}$ and a way to compute in $\F_{q^m}$.
\end{proof}

We remark that these codes are also efficiently decodable, as shown by Gabidulin~\cite{Gabidulin85a}, Delsarte~\cite{Delsarte78}, and Roth~\cite{Roth91}. 

While the above codes are thus optimal, they are only defined for \emph{finite} fields, and in particular have no analogue over the complex numbers. In fact, there is provably no analogue as in algebraically closed fields there is the following upper bound on the dimension of rank metric codes that is lower than the Singleton-type bound.

\begin{propositionwp}[Roth~\cite{Roth91}]\label{res:rank-metric_limit_roth}
	Let $\F$ be an algebraically closed field and $n\ge 1$ and $n\ge r\ge 0$.  Let $\cC\subseteq\F^{n\times n}$ be a rank-metric code with distance $\ge r+1$.  If $\cC$ is a linear code then $\dim \cC\le (n-r)^2$.  
\end{propositionwp}

Thus, while Gabidulin codes are no longer applicable in this regime, a different construction was given by Roth~\cite{Roth91} that works over any large field and meets this bound.

\begin{proposition}[Roth~\cite{Roth91}]\label{res:rank-metric_code_roth}
	Let $\F$ be a field and $m\ge n\ge 1$, $n\ge r\ge 0$ and $|\F|\ge n$. Then there is an explicit linear rank-metric code $\cR\subseteq\F^{n\times m}$ with distance $r+1$ and $\dim \cR=(n-r)(m-r)$.
\end{proposition}
\begin{proof}
	For a matrix $M\in\F^{\zr{n}\times\zr{m}}$ define the \emph{$k$-diagonal $M^{(k)}$ of $M$} to be those entries $(M_{i,j})_{i+j=k}$ where we do the addition in the integers, recalling that this notation means we index $M$ from 0.  Thus, for $0\le k<n$ there are $k+1$ entries in the $k$-diagonal of $M$, for $n\le k<m$ the diagonal $M^{(k)}$ has $n$ entries, and for $m\le k<n+m-1$ the diagonal $M^{(k)}$ has $(n+m-1)-k$ entries.

	Now recall that as $|\F|\ge n$ there is an explicit\footnote{That is, taking $d=r+1$, if $\alpha_1,\ldots,\alpha_\ell\in\F$ are distinct then $H\in\F^{\zr{d}\times \ell}$ with $H_{i,j}\eqdef \alpha_j^i$ is a Vandermonde matrix with every $k\times k$ minor of full-rank.  Thus, $\ker H\subseteq\F^\ell$ has no vectors of sparsity $\le d-1$ so that $\rspn(H)^\perp$ defines the desired code.}
	linear code $\cC_\ell\eqdef[\ell,k,r+1]_\F$ with $k=\ell-(r+1)+1$ for any $r\le \ell\le n$.  For $1\le\ell\le r$ define $\cC_\ell$ be the single codeword $\cC_\ell=\{0\}$ so that $\cC_\ell$ has distance $\ge r+1$ (trivially).

	Now take $\cR\subseteq\F^{n\times m}$ to be
	\[
		\cR\eqdef \{M \in \F^{n \times m} \st \forall 0\le k<n+m-1, M^{(k)}\in \cC_{|M^{(k)}|}\}
		\;.
	\]
	Thus, it follows that for any matrix $M\in\cR$ that all non-zero diagonals $M^{(k)}$ have at least $r+1$ non-zero entries.  

	Note that $\dim \cC_\ell=\max\{0,\ell-r\}$ for all $1\le\ell\le n$.  Thus, let us now count the dimension of $\cR$ as measured from the full-dimension.
	\begin{align*}
		nm&-\dim \cR\\
		&=\sum_{0\le k<n+m-1} |M^{(k)}|-\dim\cC_{|M^{(k)}|}\\
		&=\sum_{0\le k<n+m-1} |M^{(k)}|-\max\{0,|M^{(k)}|-r\}\\
		&=\sum_{0\le k<n+m-1} \min\{|M^{(k)}|,r\}\\
		&=\sum_{0\le k <r} \min\{|M^{(k)}|,r\}
		+\sum_{r\le k<(n+m-1)-r} \min\{|M^{(k)}|,r\}
		+\sum_{(n+m-1)-r\le k<n+m-1} \min\{|M^{(k)}|,r\}\\
		&=\sum_{0\le k< r} (k+1)
		+\sum_{r\le k<(n+m-1)-r} r
		+\sum_{(n+m-1)-r\le k<n+m-1} (n+m-1)-k\\
		&=\binom{r+1}{2}+r(n+m-1-2r)+\binom{r+1}{2}\\
		&=r(n+m-r)
		\;.
	\end{align*}
	Thus, $\dim \cR=nm-r(n+m-r)=(n-r)(m-r)$ as desired.

	Now consider a matrix $M$ with rank $\le r$.  Consider the first $k$ where the diagonal $M^{(k)}$ is non-zero.  Note that the rows with non-zero entries in this diagonal $M^{(k)}$ must be linearly independent as they form a triangular system.  Thus, $M^{(k)}$ is $r$-sparse, from which it follows that $M\notin\cR$.  Thus $\cR$ is a distance $\ge r+1$ rank-metric code. That $\cR$ is explicit is clear from construction.
\end{proof}

While Roth~\cite{Roth91} did not provide an algorithm to decode the above codes, such an algorithm was provided Forbes and Shpilka~\cite{ForbesShpilka12}.  Perhaps more relevant for this paper is that Forbes and Shpilka~\cite{ForbesShpilka12} realized that when the maximum distance separable (MDS) codes $\cC_\ell$ used in the above proof are instantiated as dual Reed-Solomon codes, one can show that the above codes actually correspond to evaluation-based codes of bivariate polynomials.  

\begin{propositionwp}[Forbes and Shpilka~\cite{ForbesShpilka12}]
	Let $\F$ be a field and $m\ge n\ge 1$ and $n\ge r\ge 0$. Let $\omega\in\F$ be an element with multiplicative order $\ge n$.  Identifying matrices $\F^{n\times m}$ with low-degree bivariate polynomials $\F[x,y]^{<n,<m}$, define
	\[
		\cC\eqdef\{ f(x,y)\st f\in\F[x,y]^{<n,<m}, \forall i\in\zr{r}, f(x,\omega^i x)=0\}
		\;.
	\]
	Then $\cC$ is an explicit linear rank-metric code with distance $r+1$ and dimension $(n-r)(m-r)$.
\end{propositionwp}

While one can prove the above by showing that they are an instantiation of Roth's~\cite{Roth91} code, Forbes and Shpilka~\cite{ForbesShpilka12} also gave a proof that we presented as the construction of a lossless two-source condenser in \autoref{res:two-src_condense-tensor_pruned} (and actually that construction can be further pruned to match the above parameters).

\section{Existential Arguments}\label{sec:prob-method}

In this section, we give arguments via the probabilistic method to show that the relevant pseudorandom objects exist.  These arguments are completely standard, but we include them to demonstrate the quantitative gaps between our explicit constructions and the existential constructions.

\subsection{Probabilistic Tools}

We begin by establishing some lemmas that detail probabilistic estimates for when a matrix is low-rank.  Some lemmas are standard (for which proofs can be found for example in the thesis of Forbes~\cite{Forbes14}) and some of which we lack a good reference and thus include for completeness. The first lemma shows how random linear maps affect full rank matrices.

\begin{lemmawp}\label{res:rand-times-full-rank}
	Let $\F$ be a finite field and $M\in\F^{n\times r}$ be of rank $r$. Let $\rE$ be a random variable uniformly distributed over matrices in $\F^{m\times n}$.  Then $\rE\cdot M$ is uniformly distributed over $\F^{m\times r}$.
\end{lemmawp}

We now give an simple estimate for the number of low-dimensional subspaces (which is notably smaller than the number of low-rank matrices).

\begin{lemmawp}\label{res:q-binom-ub-crude-sym}
	Let $\F_q$ be a finite field and for $n\ge r$, the number of dimension $r$ subspaces of $\F_q^n$ is $\le \min\{q^{rn},q^{(n-r)n}\}$. Further, this inequality is strict for $r\in (0,n)$.
\end{lemmawp}

The above estimate is good enough for many purposes, but we will need a more refined version which we now develop. We first give an estimate for the probability that a random square matrix is full rank.

\begin{lemmawp}[{see for example Forbes~\cite[Lemma D.2.4]{Forbes14}}]\label{res:count-rankk-matrices-asymp}
	Let $\F_q$ be a finite field and $n\ge r\ge 1$. The probability a random matrix in $\F_q^{r\times r}$ has rank $r$ is $>\left(1-\frac{1}{q}\right)^{\frac{q}{q-1}}$.
\end{lemmawp}

We now refine this to more amenable form using that $1+x\le \e^x$.

\begin{corollarywp}
	Let $\F_q$ be a finite field and $n\ge r\ge 1$. The probability a random matrix in $\F_q^{r\times r}$ has rank $r$ is $>\e^{-\frac{q}{(q-1)^2}}$.
\end{corollarywp}

Turning this probabilistic estimate into a counting result gives the following.

\begin{corollarywp}\label{res:counting-full-rank-matrices}
	Let $\F_q$ be a finite field and $r\ge 1$. The number of matrices $\F_q^{r\times r}$ with rank $r$ is $>\e^{-\frac{q}{(q-1)^2}} q^{r^2}$.
\end{corollarywp}

We now use this matrix counting result to count subspaces.

\begin{lemma}\label{res:q-binom-ub-sym}
	Let $\F_q$ be a finite field and $n\ge r\ge 1$. The number of $r$-dimensional subspaces of $\F_q^n$ is $<\e^{\frac{q}{(q-1)^2}}\cdot q^{r(n-r)}$.
\end{lemma}
\begin{proof}
	Consider the map from rank $r$ matrices in $\F_q^{n\times r}$ to $r$-dimensional subspaces of $\F_q^n$ defined by $M\mapsto\cspn(M)$.  This map is well-defined and surjective. Consider then the pre-image of a rank $r$ subspace $V\subseteq\F_q^n$. Each of these matrices is related by an invertible change of basis given by an invertible $r\times r$ matrix.  Thus, the pre-images of this map are of size $>\e^{-\frac{q}{(q-1)^2}} q^{r^2}$ by \autoref{res:counting-full-rank-matrices}. The claim follows by applying double-counting using that there are $q^{rn}$ matrices in $\F_q^{n\times r}$.
\end{proof}

We remark that the number of subspaces is always at least $q^{r(n-r)}$, so the above bound is asymptotically correct for large $q$. We now bound the probability a random matrix is low-rank.

\begin{lemma}\label{res:prob-matrix-low-rank}
	Let $\F_q$ be a finite field and $n,m\ge r\ge 1$.  Let $\rM$ be a random variable uniformly distributed over matrices in $\F_q^{n\times m}$.  Then,
	\[
		\Pr[\rank \rM\le r]<\min\{q^{-(nm-r(n+m))},q^{-(n-m)(m-r)},\e^{\frac{q}{(q-1)^2}}q^{-(n-r)(m-r)}\}
		\;.
	\]
\end{lemma}
\begin{proof}
	A matrix $\rM$ is rank $\le r$ iff there exists some $V\subseteq\F_q^m$ with $\dim V=r$ such that $\cspn \rM\subseteq V$.  Equivalently, there exists some $V^{\perp}\subseteq\F_q^m$ with $\dim V^\perp=m-r$ such that $(\cspn \rM)^\perp \supseteq V^\perp$.  Taking $A\in\F^{m\times(m-r)}$ to be some basis for $V^\perp$, so that $\rank A=m-r$, this is equivalent to saying that $\rM A=0$.  Thus, we have that
	\begin{align*}
		\Pr[\rank\rM\le r]
		&=\Pr[\exists V\st \cspn A=V^\perp, \rM A=0]\\
		&\le \sum_V \Pr[\rM A=0]
		\intertext{By \autoref{res:rand-times-full-rank} the random variable $\rM A$ is uniformly distributed over $\F_q^{n\times (m-r)}$,}
		&= \sum_V q^{-n(m-r)}
		\intertext{Counting such subspaces $V$ by \autoref{res:q-binom-ub-crude-sym}, using that $0<r<m$,}
		&<q^{m\cdot \min\{r,m-r\}} q^{-n(m-r)}\\
		&\le\min\{q^{-(nm-r(n+m))},q^{-(n-m)(m-r)}\}
		\;.
	\end{align*}
	Alternatively, if we count subspaces with the more refined estimate of \autoref{res:q-binom-ub-sym}, we get that
	\begin{align*}
		\Pr[\rank\rM\le r]
		&< \e^{\frac{q}{(q-1)^2}}q^{r(m-r)} \cdot q^{-n(m-r)}\\
		&=\e^{\frac{q}{(q-1)^2}} q^{-(m-r)(n-r)}
		\;.
		\qedhere
	\end{align*}
\end{proof}

Finally, as the quantity will come up more than once, we give the following approximation. It follows since for $q\ge 3$ we have that $\ln q\ge 1$ and for $q=4$ that $\frac{2q}{(q-1)^2}=\nicefrac{8}{9}<1$.

\begin{lemmawp}\label{res:approx-exp-q}
	For $q\ge 4$,
	\[
		\frac{2q}{(q-1)^2\ln q}\le 1
		\;.
		\qedhere
	\]
\end{lemmawp}

\subsection{Dimension Expanders}

We now turn to non-constructive existence of good dimension expanders. We first study when random matrices expand subspaces of a given dimension $r$ to a larger dimension $t$ (dimension expanders must expand all subspaces of dimension \emph{at most} $r$).

\begin{proposition}\label{res:dim-exp_prob-method=}
	Let $\F_q$ be a finite field. Let $n\ge t\ge r\ge 1$. Let $\rA_1,\ldots,\rA_d$ be random variables uniformly distributed over matrices in $\F_q^{n\times n}$.  Then with probability $>1-\nicefrac{1}{q^r}$, for any subspace $V\subseteq\F_q^n$ of dimension $r$ we have that
	\[
		\dim\sum_{i\in [d]} \rA_i(V) \ge t
		\;,
	\]
	assuming that
	\[
		d\ge\frac{t-1}{r}+\frac{n-r+1}{n-t+1}+\frac{2q}{(q-1)^2\ln q}\\
		\;.
	\]
	In particular, if $q\ge 4$ then it suffices for
	\[
		d\ge\frac{t-1}{r}+\frac{n-r+1}{n-t+1}+1
		\;.
	\]
\end{proposition}
\begin{proof}
	Fix some subspace $V\subseteq\F_q^n$ of dimension $r$, and let $M\in\F_q^{n\times r}$ be a rank $r$ matrix with $\cspn M=V$.  Then we see that $\dim\sum_i \rA_i(V) \ge t$ iff the $\F_q^{n\times rd}$ block matrix
	\[
		\rN(V)
		\eqdef
		\begin{bmatrix}
			\rA_1 M | \cdots |\rA_d M
		\end{bmatrix}
	\]
	has rank $\ge t$.  As $M$ has full rank and the $\rA_i$ are uniformly random, it follows from \autoref{res:rand-times-full-rank} that $\rN(V)$ is uniformly random matrix over $\F^{n\times rd}$.  Thus, by the refined version of \autoref{res:prob-matrix-low-rank},
	\begin{align*}
		\Pr[\rank \rN(V) < t]
		&=\Pr[\rank \rN(V)\le t-1]\\
		&<\e^{\frac{q}{(q-1)^2}}q^{-(n-t+1)(rd-t+1)}
	\end{align*}
	Thus, applying a union bound,
	\begin{align*}
		\Pr[\exists V \st \rank \rN(V) < t]
		&\le \sum_V \Pr[\rank \rN(V)\le t-1]\\
		&<\sum_V\e^{\frac{q}{(q-1)^2}}q^{-(n-t+1)(rd-t+1)}
		\intertext{counting subspaces by \autoref{res:q-binom-ub-sym},}
		&\le \e^{\frac{q}{(q-1)^2}}q^{r(n-r)}\cdot \e^{\frac{q}{(q-1)^2}}q^{-(n-t+1)(rd-t+1)}
	\end{align*}
	This quantity is $\le q^{-r}$ iff
	\[
		(n-t+1)(rd-t+1)\ge r(n-r+1)+\frac{2q}{(q-1)^2\ln q}
		\;.
	\]
	Dividing by $r(n-t+1)$ on both sides, this is equivalent to
	\begin{align*}
		d&\ge \frac{t-1}{r}+\frac{n-r+1}{n-t+1}+\frac{2q}{r(n-t+1)(q-1)^2\ln q}
		\intertext{using that $r\ge 1$ and $n-t+1\ge 1$, it thus suffices for}
		d&\ge \frac{t-1}{r}+\frac{n-r+1}{n-t+1}+\frac{2q}{(q-1)^2\ln q}
		\;,
	\end{align*}
	as desired. Appealing to \autoref{res:approx-exp-q} yields the claim for $q\ge 4$.
\end{proof}

We now apply a union bound to the above to existentially obtain dimension expanders.

\begin{proposition}\label{res:dim-exp_prob-method}
	Let $\F_q$ be a finite field, $n\ge 1$, $\eps>0$ and $\alpha\in \R$ with $1\le\alpha<\nicefrac{1}{\eps}$. Then there exist a collection matrices $\cA=\{A_1,\ldots,A_d\}\subseteq\F^{n\times n}$ which is a $(\eps,\alpha)$-dimension expander of degree $d$ whenever
	\[
		d\ge \alpha+\frac{1}{1-\alpha\eps}+\frac{2q}{(q-1)^2\ln q}
		\;.
	\]
	In particular, if $q\ge 4$ then it suffices for
	\[
		d\ge \alpha+\frac{1}{1-\alpha\eps}+1
		\;.
		\qedhere
	\]
\end{proposition}
\begin{proof}
	For $\cA$ to be $(\eps,\alpha)$ expander means that for subspaces $V\subseteq\F^n$ of dimension $r\le \eps n$ we need that $\dim\sum_i A_i(V)\ge \alpha\dim V$.  That is, $\dim\sum_i A_i(V)\ge \ceil{\alpha r}$.  For any fixed $r\le \eps n$, \autoref{res:dim-exp_prob-method=} shows that a random collection $\cA$ will have this property for all such $V$ with probability $>1-\nicefrac{1}{q^r}$ as long as
	\begin{align*}
		d&\ge\frac{\ceil{\alpha r}-1}{r}+\frac{n-(r-1)}{n-(\ceil{\alpha r}-1)}+\frac{2q}{(q-1)^2\ln q}
		\intertext{As $\ceil{\alpha r}-1\le \alpha r$ it follows that it is sufficient for}
		d&\ge\alpha+\frac{n}{n-\alpha r}+\frac{2q}{(q-1)^2\ln q}
		\intertext{and thus sufficient for}
		d&\ge\alpha+\frac{1}{1-\alpha \eps}+\frac{2q}{(q-1)^2\ln q}
		\;,
	\end{align*}
	where this bound is independent of $r$. 

	Now, taking the union bound over all $1\le r\le \eps n$ we see that the failure probability is at $<\sum_{r=1}^\infty \nicefrac{1}{q^r}<1$, so it follows that such a collection $\cA$ exists that expands by a factor of $\alpha$ each subspace of dimension of at most $\le \eps n$.  For $q\ge 4$ we can appeal to the latter part of \autoref{res:dim-exp_prob-method=}.
\end{proof}

\subsection{Lossy Rank Condensers}

We now give an argument that good lossy rank condensers exist.

\begin{proposition}\label{res:lossy-seeded_prob-method=}
	Let $\F_q$ be a finite field. Let $n\ge r\ge 1$, $\eps\ge 0$ and $t>(1-\eps)r$. Let $\rE_1,\ldots,\rE_k$ be random variables uniformly distributed over matrices in $\F_q^{t\times n}$.  Then with probability $>1-\nicefrac{1}{q^r}$, $\cE\eqdef\{\rE_1,\ldots,\rE_k\}$ is a $(r,\eps)$-lossy rank condenser whenever
	\[
		k\ge\frac{rn+\frac{q}{(q-1)^2\ln q}}{(t-(1-\eps)r)(\floor{\eps r}+1)-\frac{q}{(q-1)^2\ln q}}
		\;,
	\]
	whenever the denominator is positive. In particular, if $q\ge 4$ then it suffices for
	\[
		k\ge\frac{rn+1}{(t-(1-\eps)r)(\floor{\eps r}+1)-1}
		\;,
	\]
	whenever the denominator is positive.
\end{proposition}
\begin{proof}
	We bound the probability that $\cE$ is not such a condenser. Fix some matrix $M\in\F_q^{n\times r}$ of rank $r$.  We then bound the probability
	\begin{align*}
		\Pr_{\rE_i}[\forall i\in[k], \rank \rE_i M<(1-\eps)\rank M]
		&=\Pr_{\rE_i}[\forall i\in[k], \rank \rE_i M\le r-\floor{\eps r}-1]
		\intertext{using independence of the $\rE_i$,}
		&=\left(\Pr_{\rE_1}[\rank \rE_1 M\le r-\floor{\eps r}-1]\right)^k
		\intertext{as $\rE_1 M$ is uniformly distributed over $\F_q^{t\times r}$ (by \autoref{res:rand-times-full-rank}), and invoking the finer form of \autoref{res:prob-matrix-low-rank},}
		&<\left(\e^{\frac{q}{(q-1)^2}}q^{-(t-(r-\floor{\eps r}-1))(r-(r-\floor{\eps r}-1))}\right)^k
		\;.
	\end{align*}
	Union bounding over all such $M$ (and using \autoref{res:q-binom-ub-sym} to count such $M$),
	\begin{align*}
		\Pr[\cE \text{ not a $(r,\eps)$-condenser}]
		&=\Pr[\exists M\st \forall i, \rank \rE_i M<(1-\eps)\rank M]\\
		&<\e^{\frac{q}{(q-1)^2}}q^{r(n-r)}\cdot \left(\e^{\frac{q}{(q-1)^2}}q^{-(t-(r-\floor{\eps r}-1))(r-(r-\floor{\eps r}-1))}\right)^k
		\;.
	\end{align*}
	This quantity is at most $q^{-r}$ iff 
	\begin{align*}
		k&\ge\frac{r(n-r+1)+\frac{q}{(q-1)^2\ln q}}{(t-(r-\floor{\eps r}-1))(\floor{\eps r}+1)-\frac{q}{(q-1)^2\ln q}}
		\intertext{so that it is sufficient (as $r\ge 1$, and $\floor{\eps r}+1\ge \eps r$) that}
		k&\ge\frac{rn+\frac{q}{(q-1)^2\ln q}}{(t-(1-\eps)r)(\floor{\eps r}+1)-\frac{q}{(q-1)^2\ln q}}
		\;,
	\end{align*}
	whenever this denominator is positive. Appealing to \autoref{res:approx-exp-q} yields the claim for $q\ge 4$.
\end{proof}

We note that the finer form of \autoref{res:prob-matrix-low-rank} was not strictly needed in the above to obtain non-trivial results, but this finer form allows us to see that the explicit construction of \autoref{res:gk13-lossy-condenser} is slightly suboptimal as it requires $t\ge r$ (and this sub-optimality manifests in the constructions of dimension expanders as discussed after \autoref{res:tensor-then-condense_instantiate_gamma0}).

We now observe that by the above we can find a collection $\cE$ that is a $(\le r,\eps)$-condenser.  As any $(r,0)$-lossy condenser is also a $(s,0)$-condenser for all $s\le r$ (\autoref{res:lossless_le-v-eq}), this is only non-trivial for $\eps>0$.  Note that this is indeed non-trivial, as there are $(3n,\nicefrac{2}{2})$-lossy condensers on $\F^{4n}$ that are not $(s,\nicefrac{2}{3})$-condensers for any $s\le n$ (see \autoref{res:dim-le-r_dim-eq-r}).

\begin{proposition}\label{res:lossy-seeded_prob-method}
	Let $\F_q$ be a finite field. Let $n\ge r\ge 1$, $\eps>0$ and $t>(1-\eps)r$.   Then there is a collection $\cE$ of matrices $\cE\subseteq\F_q^{t\times n}$ that is a $(\le r,\eps)$-lossy rank condenser whenever
	\[
		k\ge\frac{n+\frac{q}{(q-1)^2\ln q}}{\eps (t-(1-\eps)r)-\frac{q}{(q-1)^2\ln q}}
		\;,
	\]
	whenever this denominator is positive. In particular, if $q\ge 4$ then it suffices for
	\[
		k\ge\frac{n+1}{\eps (t-(1-\eps)r)-1}
		\;,
	\]
	whenever this denominator is positive.
\end{proposition}
\begin{proof}
	Take $\cE\eqdef\{\rE_1,\ldots,\rE_k\}$ where $\rE_i$ are uniformly and independently distributed over $\F_q^{t \times n}$. By \autoref{res:lossy-seeded_prob-method=} we see that $\cE$ is a $(s,\eps)$-condenser with probability $>1-q^s$ as long as
	\begin{align*}
		k&\ge\frac{sn+\frac{q}{(q-1)^2\ln q}}{(t-(1-\eps)s)(\floor{\eps s}+1)-\frac{q}{(q-1)^2\ln q}}
		\intertext{as $\floor{\eps s}+1\ge \eps s$ and $s\le r$, it is then sufficient for}
		k&\ge\frac{sn+\frac{q}{(q-1)^2\ln q}}{(t-(1-\eps)r)\eps s-\frac{q}{(q-1)^2\ln q}}
		\intertext{whenever this denominator is positive. Simplifying further, we see that}
		k&\ge\frac{n+\frac{q}{s(q-1)^2\ln q}}{\eps (t-(1-\eps)r)-\frac{q}{s(q-1)^2\ln q}}
		\intertext{and thus as $s\ge 1$, that}
		k&\ge\frac{n+\frac{q}{(q-1)^2\ln q}}{\eps (t-(1-\eps)r)-\frac{q}{(q-1)^2\ln q}}
		\;,
	\end{align*}
	is sufficient whenever this denominator is positive, where this last bound is independent of $s$.  Thus, with this value of $k$ we see that $\cE$ is a $(s,\eps)$-condenser for all $1\le s\le r$ except with probability $<\sum_{s=1}^r q^{-s}\le \sum_{s=1}^\infty q^{-s}<1$.  Thus, such a condenser $\cE$ exists, where for $q\ge 4$ we appeal to the latter half of \autoref{res:lossy-seeded_prob-method=}.
\end{proof}

\subsection{Two-Source Rank Condensers}

We now give an argument that good two-source rank condensers exist.

\begin{proposition}\label{res:two-src_prob-method}
	Let $\F_q$ be a finite field. Let $n\ge r\ge 1$, $m\ge s\ge 1$ and $\eps\ge0$.  Let $\rE$ be a random variable uniformly distributed over matrices in $\F_q^{t\times nm}$.  Then $f:\F^n\times\F^m\to\F^t$ defined by $f(\vv,\vw)=E\cdot(\vv\otimes\vw)$ is a bilinear $(r,s,\eps)$-two-source rank condenser with probability $>1-\nicefrac{1}{q^r}$, assuming that
	\[
			t
			\ge
			\frac{n}{\eps s}+\frac{m}{\eps r}+(1-\eps) rs+\frac{2q}{(q-1)^2\ln q}
			\;.
	\]
	for $\eps>0$. If $\eps=0$, then $f$ is such a $(r,s,0)$-condensers with probability $>1-\frac{1}{q^r}$, assuming that
	\[
		t\ge rn+sm+rs+\frac{2q}{(q-1)^2\ln q}-1
		\;.
	\]
\end{proposition}
\begin{proof}
	Note that $f$ is indeed a bilinear function as seen from the equivalence of the two definitions (\autoref{res:bilinear-cond_alt}). That $f$ fails to be a $(r,s,\eps)$-condenser means that there are full-rank matrices $A\in\F_q^{n\times r}$ and $B\in\F_q^{m\times s}$ such that $f(A,B)=\rE\cdot (A\otimes B)$ has $\rank f(A,B)<(1-\eps)rs$, that is $\rank f(A,B)\le rs-(\floor{\eps rs}+1)$.  Thus, appealing to the above lemmas (and that $A\otimes B$ is rank $rs$),
	\begin{align*}
		\Pr[f \text{ is not a $(r,s,\eps)$-condenser}]
		&=\Pr[\exists A,B\st \rank f(A,B)\le rs-(\floor{\eps rs}+1)]\\
		&\le \sum_{A,B} \Pr[\rank f(A,B)\le rs-(\floor{\eps rs}+1)]\\
		&=\sum_{A,B} \Pr[\rank (\rE\cdot(A\otimes B)) \le rs-(\floor{\eps rs}+1)]
		\intertext{appealing to \autoref{res:rand-times-full-rank} to see that $\rE\cdot (A\otimes B)$ is a random $t\times rs$ matrix, and then applying the third form of \autoref{res:prob-matrix-low-rank}, as well as \autoref{res:q-binom-ub-crude-sym},}
		&<\sum_{A,B} \e^{\frac{q}{(q-1)^2}} q^{-(t-(rs-(\floor{\eps rs}+1)))(rs-(rs-(\floor{\eps rs}+1)))}\\
		&< \e^{\frac{q}{(q-1)^2}} q^{r(n-r)} \cdot q^{sm} \cdot \e^{\frac{q}{(q-1)^2}} q^{-(t-(rs-(\floor{\eps rs}+1)))(\floor{\eps rs}+1)}\\ 
	\end{align*}
	The above quantity is $\le q^{-r}$ iff
	\begin{align*}
		t&\ge\frac{r(n-r+1)+sm}{\floor{\eps rs}+1}+(rs-\floor{\eps rs}+1)+\frac{2q}{(\floor{\eps rs}+1)(q-1)^2\ln q}
		\intertext{using that $r,s\ge 1$, it suffices for}
		t&\ge\frac{rn+sm}{\floor{\eps rs}+1}+(rs-\floor{\eps rs}-1)+\frac{2q}{(q-1)^2\ln q}
		\intertext{which yields the result for $\eps=0$. For $\eps>0$, we use that  $\eps rs\le \floor{\eps rs}+1$, so see that it suffices for}
		t&\ge\frac{rn+sm}{\eps rs}+(1-\eps)rs+\frac{2q}{(q-1)^2\ln q}
		\;.
		\qedhere
	\end{align*}
\end{proof}

We now use the above to derive a lossy two-source condenser that works even when the first source is \emph{small}. Note that as before this is only interesting when $\eps>0$, as when $\eps=0$ condensers automatically work for smaller ranks (\autoref{res:two-src_lossless_le}).

\begin{proposition}\label{res:two-src_prob-method<}
	Let $\F_q$ be a finite field. Let $n\ge r\ge 1$, $m\ge s\ge 1$ and $\eps>0$. Then there exists a $f:\F^n\times\F^m\to\F^t$ which is a bilinear $(\le r,s,\eps)$-two-source rank condenser, assuming that
	\[
		t
		\ge
		\frac{n}{\eps s}+\frac{m}{\eps }+(1-\eps) rs+\frac{2q}{(q-1)^2\ln q}
		\;.
	\]
\end{proposition}
\begin{proof}
	Let $\rE$ be a random variable uniformly distributed over matrices in $\F_q^{t\times nm}$ and define $f$ by $f(\vv,\vw)=E\cdot(\vv\otimes\vw)$.  For $r'$ with  $1\le r'\le r$, \autoref{res:two-src_prob-method} yields that $f$ is a bilinear $(r',s,\eps)$-two-source condenser with probability $>1-\nicefrac{1}{q^{r'}}$, assuming that
	\begin{align*}
		t&\ge \frac{n}{\eps s}+\frac{m}{\eps r'}+(1-\eps) r's+\frac{2q}{(q-1)^2\ln q}
		\intertext{in particular, as $1\le r'\le r$, if}
		t&\ge \frac{n}{\eps s}+\frac{m}{\eps}+(1-\eps) rs+\frac{2q}{(q-1)^2\ln q}
		\;,
	\end{align*}
	where this last bound is independent of $r'$. The probability of failure, union bounding over all $1\le r'\le r$, is at most $\sum_{r'=1}^\infty q^{-r'}<1$.  Thus the desired $f$ exists.
\end{proof}

\end{document}